\documentclass[journal]{IEEEtran}
\usepackage[left=0.7in,right=0.7in,top=0.6in,bottom=0.6in]{geometry}

\usepackage{mathrsfs}
\usepackage{latexsym}
\usepackage{graphicx}
\usepackage{epsfig}
\usepackage{subfigure}
\usepackage{array}
\usepackage{amsmath}
\usepackage{amssymb}
\usepackage{color, soul}
\usepackage{amsthm}
\usepackage{enumerate}
\usepackage{algpseudocode}
\usepackage{algorithm}
\usepackage{bm}

\theoremstyle{plain} %default
\newtheorem{thm}{Theorem}

\newtheorem{lemma}{Lemma}
\newtheorem{prop}{Proposition}
\theoremstyle{definition}

\DeclareMathOperator*{\argmax}{arg\,max}

\def\bal#1\eal{\begin{align}#1\end{align}}

\newcommand{\bxi} {\boldsymbol{\xi}}

\newcommand{\bW}{{\bf W}}
\newcommand{\bH}{{\bf H}}

\newcommand{\bR}{{\bf R}}

\newcommand{\bK}{{\bf K}}
\newcommand{\bI}{{\bf I}}

\newcommand{\bM}{{\bf M}}
\newcommand{\bN}{{\bf N}}
\newcommand{\bA}{{\bf A}}

\newcommand{\bD}{{\bf D}}
\newcommand{\bQ}{{\bf Q}}
\newcommand{\bZ}{{\bf Z}}

\newcommand{\bT}{{\bf T}}

\newcommand{\bLam}{{\bf \Lambda}}

\newcommand{\by}{{\bf y}}
\newcommand{\bx}{{\bf x}}

\newcommand{\ba}{{\bf a}}

\newcommand{\br}{{\bf r}}
\newcommand{\bz}{{\bf z}}
\newcommand{\bo}{{\bf 0}}
\newcommand{\bk}{{\bf k}}
\newcommand{\bw}{{\bf w}}
\newcommand{\bRequire}{{\bf Require}}
\newcommand{\sS}{\mathcal{S}}
\newcommand{\bp} {\begin{proof}}
\newcommand{\ep} {\end{proof}}

\newcommand{{\Rb}} {\right)}

\newcommand{{\Rf}} {\right\}}
\begin{document}

\title{Algorithms for Globally-Optimal Secure Signaling over Gaussian MIMO Wiretap Channels Under Interference Constraints}

\author{Limeng Dong, Sergey Loyka, Yong Li

\vspace*{-1\baselineskip}

	\thanks{This paper was presented in part at the 5th IEEE Global Conference on Signal and Information Processing, Montreal, Canada, Nov. 2017 \cite{Dong-17b}.} 	
	\thanks{L. Dong is with the Ministry of Education Key Lab for Intelligent Networks  and Network Security, School of Information and Communications Engineering, Xi'an Jiaotong University, Xi'an, Shaanxi, 710049, China, e-mail: dlm\_nwpu@hotmail.com. This work was done when L. Dong was visiting the School of Electrical Engineering and Computer Science, University of Ottawa, Canada.}
	\thanks{S. Loyka is with the School of Electrical Engineering and Computer Science, University of Ottawa, Canada, e-mail: sergey.loyka@uottawa.ca.}
	\thanks{Y. Li is with the School of Electronics and Information, Northwestern Polytechnical University, Xi’an, China, e-mail: ruikel@nwpu.edu.cn.}
    %\thanks{Corresponding authors: S. Loyka and Y. Li.}
}
\maketitle

%\vspace*{-1\baselineskip}
%======================================================================
\begin{abstract}

Multi-user Gaussian MIMO wiretap channel is considered under interference power constraints (IPC), in addition to the total transmit power constraint (TPC). Algorithms for \textit{global} maximization of its secrecy rate are proposed. Their convergence to the secrecy capacity is rigorously proved and a number of properties are established analytically. Unlike known algorithms, the proposed ones are not limited to the MISO case and are proved to converge to a \textit{global} rather than local optimum in the general MIMO case, even when the channel is not degraded. In practice, the convergence is fast as only a small to moderate number of Newton steps is required to achieve a high precision level. The interplay of TPC and IPC is shown to result in an unusual property when an optimal point of the max-min problem does not provide an optimal transmit covariance matrix in some (singular) cases. To address this issue, an  algorithm is developed to compute an optimal transmit covariance matrix in those singular cases. It is shown that this algorithm also solves the dual (nonconvex) problems of \textit{globally} minimizing the total transmit power subject to the secrecy and interference constraints; it provides the minimum transmit power and respective signaling strategy needed to achieve the secrecy capacity, hence allowing power savings.

\end{abstract}

%\vspace*{-1\baselineskip}
%======================================================================
\section{Introduction}

Ever-growing number of wireless users and their traffic, open system architectures and aggressive frequency re-use as well as operation in unlicensed bands envisioned in 5G systems \cite{Shafi-17} create significant potential for inter-user interference, which needs to be carefully controlled and mitigated. Multiple antennas offer a significant potential for doing so in the space domain, especially in the context of massive MIMO \cite{Lu-14}. This approach to interference mitigation and control has been investigated earlier in the context of cognitive radio (CR) \cite{Zhang-08}, where secondary users are allowed to use the same bandwidth as primary users (who are the license holders) but are required to cause no significant interference to them. On the other hand, open system architectures and co-existence of several users in the same bandwidth in combination with the broadcast nature of wireless channels make transmissions vulnerable to eavesdropping of confidential information (e.g. e-commerce and e-health, mobile banking, Internet transactions, etc.) so that some form of secrecy protection is needed. In this context, physical-layer security approach has emerged as a valuable complement to the traditional cryptography-based approach for modern wireless networks \cite{Bloch-11}-\cite{Wu-18}. In this approach, the secrecy of communications is ensured at the physical layer by exploiting the properties of wireless communication channels so that no transmitted information can be recovered by malicious eavesdroppers. Wiretap channel (WTC) is widely used as a model of secrecy communications and its secrecy capacity became the key metric of performance \cite{Bloch-11}-\cite{Oggier-10}.

\subsection{Literature review}

Using this approach in combination with MIMO systems offers significant new opportunities for enhancing the secrecy of multi-user wireless systems via space-domain processing. The MIMO WTC model became a popular tool to study physical-layer security, where the transmitter (Tx) sends confidential information to the receiver (Rx) while an eavesdropper (Ev) observes the transmission. The main performance metric, which is an ultimate upper bound to reliable and secret communications, is the secrecy capacity, defined operationally as the maximum achievable rate on the Tx-Rx link subject to the reliability (low error probability) and secrecy (low information leakage on the Tx-Ev link) criteria \cite{Bloch-11}-\cite{Wu-18}. The secrecy capacity of Gaussian MISO (multiple-input single-output) WTC has been established in \cite{Khisti-10a} and  further extended to the full MIMO case in \cite{Khisti-10b}\cite{Oggier-10}, where the optimality of Gaussian signaling has also been established.

Hence, finding the secrecy capacity amounts to finding an optimal input (transmit) covariance matrix. This problem is still open analytically in the general case since the underlying optimization problem is not convex and hence very hard to solve, either numerically or analytically, while some special cases (MISO, full-rank MIMO, rank-1 MIMO, weak eavesdropper, identical right singular vectors of Rx and Ev channels, etc.) have been solved \cite{Khisti-10a}-\cite{Loyka-17}. The two Tx antennas case was studied in details in \cite{Vaezi-17}, the massive MIMO setting was considered in \cite{Basciftci-18}\cite{Asaad-18}, and finite-alphabet signaling was also studied \cite{Aghdam-19}; an overview of recent results can be found in \cite{Regalia-15}\cite{Wu-18}\cite{Aghdam-19}.

The Gaussian MISO WTC with multiple eavesdroppers was considered in \cite{Zhang-10}, where the original non-convex problem was transformed to a quasi-convex one that can be solved as a sequence of convex feasibility problems using the bisection method. This MISO case with multiple Evs was also studied in \cite{Li-11}, including a deterministic channel uncertainty model, where the original non-convex problem was transformed to a convex semi-definite one using determinantal inequality and the fact that optimal covariance is of rank-1 so that the inequality becomes equality; this new problem can be efficiently solved using existing convex solvers. The MISO channel with multi-eavesdroppers and stochastic channel uncertainty was studied in \cite{Chu-16}. Unfortunately, the multi-Ev studies above did not establish the optimality of Gaussian signaling (but rather assumed it) and they cannot be extended to the full MIMO case since there exists no equivalent scalar channel anymore and optimal covariance is not necessarily rank-1. New approaches are needed. In this respect, the multi-Ev case, where eavesdroppers are not cooperative, is equivalent to a compound WTC whose operational capacity was established in \cite{Liang-09}\cite{Schaefer-15} by demonstrating that Gaussian signaling is optimal. However, an optimal Tx covariance matrix is not known in the general case either.

Since there is no closed-form solution in the general case, even for a single eavesdropper, a number of numerical algorithms have been developed to maximize secrecy rates \cite{Li-13}-\cite{Cumanan-14}. As the original problem is non-convex, these algorithms use some form of convexification, where the non-convex part of the objective (the Ev part) is expanded in a Taylor series and only first two terms are kept (i.e. the non-convex part is linearized, either explicitly or implicitly) \cite{Li-13}-\cite{Cumanan-14}. Then, the approximated but convex problem is solved, an expansion point is iteratively updated and the process is repeated. The fundamental difficulty with this approach is that, even if the algorithm can be proved to converge, a convergence point is just a Karush-Kuhn-Tucker (KKT) point, but, due to the non-convex nature of the original (not approximate) problem, the KKT conditions are not sufficient for \textit{global} optimality. Hence, a convergence point of these algorithms can be a local rather than global maximum, an inflection point, a local or even global minimum \cite{Murty-87}-\cite{Nesterov-04}. All these algorithms lack provable convergence to a \textit{global} optimum due to the non-convex nature of the original problem and no way is known to overcome this fundamental difficulty. Furthermore, the gap to a global optimum is not known either.

Using a different approach, an algorithm with provable convergence to a \textit{global} optimum in the general MIMO case (with single eavesdropper) was proposed in \cite{Loyka-15}. The key idea was to  avoid any form of convexification/approximation or alternating optimization (for which proving convergence to a global optimum is out of reach), but rather to use the max-min formulation in \cite{Khisti-10b}\cite{Oggier-10}, without any approximations.
However, this algorithm cannot be used in interference-constrained environments (e.g. CR) due to three fundamental issues: (i) while the feasible set is isotropic under the Tx power constraint (TPC) alone (no limits on eigenvectors, only on the sum of eigenvalues of the Tx covariance matrix), it is not isotropic anymore when interference power constraints (IPC) are added and this has a dramatic impact on the KKT conditions and numerical algorithms used to solve them; (ii) any of the constraints, including TPC, can be inactive under IPCs while the TPC is \textit{always} active without IPC; furthermore, it is not known in advance which constraint is active and which is not so that an algorithm is required to determine this automatically; finally, (iii) a global convergence proof must include the interference  constraints and the fact that some of them may be inactive.

An interference-constrained Gaussian MISO WTC with a single Ev was studied in \cite{Pei-10}. It was shown that Gaussian signalling is optimal and the operational secrecy capacity can be expressed as a quasi-convex optimization problem, which can be subsequently reduced to a sequence of convex feasibility problems \cite{Pei-10} and they can be further solved using existing convex solvers. An imperfect channel state information (CSI) was accounted for in \cite{Pei-11}. Secrecy rate maximization of interference-constrained MISO (single-antenna Rx) WTC under single or multiple non-cooperative Evs and various channel assumptions (fixed, quasi-static or ergodic fading with full or partial channel state information) was studied in \cite{Wang-14}-\cite{Zhang-18}; artificial noise and various beamforming solutions were proposed to maximize the secrecy rate.

However, it is not known whether these solutions are optimal, i.e. achieve the secrecy capacity, and what is the actual gap to the capacity. In addition, all these studies are limited to the MISO case, i.e. single-antenna receivers, and cannot be extended to the full MIMO case due to the fundamental limitations of the approach they use, i.e. transforming a MISO channel into an equivalent scalar channel and reducing (or relaxing) the original non-convex problem to a convex or quasi-convex one. In the full MIMO case, there is no equivalent scalar channel, beamforming is not an optimal strategy in general, the original problem is not convex and it is not known how to transform it into an equivalent convex or quasi-convex problem.

%\vspace*{-1\baselineskip}
\subsection{Contributions}

Thus, a new approach is needed to deal with the full MIMO WTC under interference constraints. Unlike the previous studies in \cite{Pei-10}-\cite{Zhang-18}, in this paper we target capacity-achieving signaling over the full MIMO WTC under interference constraints and to this end develop an algorithm with provable convergence to a global optimum.

The proposed  algorithm is based on the max-min secrecy capacity characterization originally developed in \cite{Khisti-10b}\cite{Oggier-10} under the TPC alone and later extended to the joint constraints (TPC+IPC) in \cite{Dong-18}, where the optimality of Gaussian signaling was established in the general MIMO case under the joint constraints and the max and max-min secrecy capacity characterizations of \cite{Khisti-10b}\cite{Oggier-10} were shown to hold as well (even though the feasible set under the joint constraints is not isotropic).

However, no analytical solutions to either the max or max-min problems above are known in the general case (some special cases have been solved in \cite{Dong-18}\cite{Loyka-18}, but the general case remains an open problem). No algorithmic solution with provable convergence to global optimum under the joint constraints is known either. Therefore, a numerical algorithm is needed to solve the problems and thus to find an optimal Tx covariance matrix and the secrecy capacity. Such algorithm is proposed in the present paper. The distinct features of this new algorithm are that (i) it finds \textit{globally}-optimum (i.e. capacity-achieving) transmit covariance matrix, (ii) its convergence to a \textit{global} (rather than local) optimum is rigorously proved, and (iii) it is of polynomial complexity.

It should be emphasized that the standard max-only formulation of the secrecy rate maximization problem, which is dominant in the current literature, see e.g. \cite{Li-13}-\cite{Cumanan-14}, does not allow one to build an algorithm with \textit{guaranteed} convergence to a \textit{global} optimum in the general case due to the lack of problem's convexity (which makes provable global convergence out of reach, see e.g. \cite{Boyd-04}\cite{Murty-87}-\cite{Nesterov-04}). Algorithms based on the max-only formulation face a fundamental difficulty since they may get trapped in a local optimum and hence their performance may be rather poor. For example, we show in Fig. \ref{fig.R*=not.R'} that the Taylor expansion-based sub-optimal algorithm as in \cite{Cumanan-14} does get trapped at local optima (or stationary points), far away from the global one, resulting in poor performance and hence should be used with caution (or avoided at all) when the original problem is not convex. In general, non-convex problems are NP-hard (of exponential complexity) and the best one can hope for is convergence to a stationary point, which can be a local (rather than global) maximum, an inflection point or even a local minimum \cite{Murty-87}-\cite{Nesterov-04}. A convergence point may also depend on initial (starting) point, so that some bad initial points may result in bad results (e.g. a local minimum rather than maximum). The only known exception to this is the MISO case, where problem re-formulation is possible to a quasi-convex or some other tractable form, but this re-formulation is not possible in the general MIMO case (since there exists no equivalent scalar channel). On the other hand, the max-min characterization of \cite{Khisti-10b}\cite{Oggier-10}\cite{Dong-18}, while appearing to be more complicated due to two conflicting optimizations, is in fact more tractable due to its convex-concave nature.

In this paper, we use the max-min characterization to construct an  algorithm with \textit{provable} convergence to a \textit{global} optimum in the general MIMO case. This algorithm includes three key components: (i) the residual-form Newton method, (ii) the barrier method and (iii) backtracking line search. The barrier method is needed to absorb inequality constraints into the objective function, while the residual-form Newton method, in combinations with backtracking line search, generates a sequence of points which converge to a globally-optimal max-min point for which the objective value is the secrecy capacity. When combined properly, they are proved  to converge to a globally-optimal solution with any desired accuracy. In practice, only a small to moderate number of Newton steps is needed to achieve a high precision level.

While the algorithm above computes the secrecy capacity via a saddle-point of the max-min problem, its optimal covariance matrix is not necessarily a maximizer of the secrecy rate and hence cannot be used for globally-optimal (i.e. capacity-achieving) signaling in the general case. This unusual effect is entirely due to the interplay between TPC and IPC and cannot be found under the TPC alone, as in \cite{Khisti-10b}\cite{Oggier-10}\cite{Loyka-15}. To address this issue, we establish general properties of the secrecy capacity as a function of Tx power and, based on it, develop an iterative bisection algorithm in Section \ref{sec.algorithm2} (Algorithm 2), which evaluates numerically an optimal covariance in the general case with any desired accuracy and prove its convergence. Numerical experiments show that the proposed algorithms converge fast in practice and achieve higher secrecy rates (significantly higher when the channel is not degraded and its negative eigenmode is dominant) than the known sub-optimal algorithms.

Motivated by energy efficiency issues, dual problems of minimizing globally the total transmit power subject to secrecy and interference power constraints are considered in Section \ref{sec.dual}. Since these problems are not convex, standard tools of convex optimization do not apply and they are difficult to solve (where "solve" means finding global rather than local optimum). Yet, Proposition \ref{prop.dual} shows that  Algorithm 2 solves these problems as well. This provides the globally-minimum Tx power and respective signaling strategy needed to achieve a target secrecy rate under interference constraints.

Collectively, the two proposed algorithms evaluate the secrecy capacity and globally-optimal signaling strategy to achieve it in the interference-constrained multi-user Gaussian MIMO wiretap channel in the general case. This is substantially different from the known algorithms in \cite{Li-13}-\cite{Cumanan-14}\cite{Pei-10}\cite{Pei-11}, which either operate over a MISO channel only or which converge to a stationary (KKT) point only, which can be a local rather than global maximum, an inflection point, a local or even global minimum, and for which a proof of convergence to a \textit{global} optimum is out of reach.

The rest of the paper is organized as follows. Section \ref{sec.Ch.Model} introduces the channel model and gives its operational secrecy capacity; the model is general enough to include per-antenna power constraints as well. The algorithm for global maximization of secrecy rates over interference-constrained multi-user Gaussian MIMO wiretap channel is developed in Section \ref{sec.algorithm} and its convergence is rigorously proved. Based on this algorithm, Section \ref{sec.algorithm2} presents a bisection-based algorithm to evaluate numerically an optimal Tx covariance with any desired accuracy in the general case and its convergence is proved. Dual problems of minimizing the total transmit power subject to secrecy and interference power constraints are considered in Section \ref{sec.dual} and Algorithm 2 is shown to solve these problems as well.  Finally, numerical experiments to illustrate algorithms' performance and practical convergence are given in Section \ref{sec.examples}.

\emph{Notations}: bold lower-case letters ($\ba$) and capitals ($\bA$) denote vectors and matrices respectively; $\bA \ge \bo$ denotes positive semi-definite matrix $\bA$; $\bA^T$  is transposition while $\bA^+$ is Hermitian conjugation; $tr(\bA)$ is the trace; $vec(\bA)$ is the vector obtained by stacking all columns of matrix $\bA$ on top of each other and $veh(\bA)$ is the vector obtained by vectorizing only the lower triangular part of $\bA$; $\emph{diag}(\bA)$ is a diagonal matrix with the same diagonal entries as in $\bA$; $E\left \{ \cdot  \right \}$ is a statistical expectation; $\otimes$ is the Kronecker product;  $|\ba |$ and $|\bA|$ are the Euclidian norm of vector $\ba$ and determinant of matrix $\bA$; $\bI$ is the identity matrix of appropriate size.

%\newpage
%======================================================================
\section{Channel Model and Secrecy Capacity}
\label{sec.Ch.Model}
Let us consider the standard Gaussian MIMO WTC model as shown in Fig. \ref{fig.1},  where the transmitter (Tx) sends confidential information to the receiver (Rx) while $N$ eavesdroppers (Ev), who may be just other users in a multi-user system, intercept the transmission; the Evs are assumed to be cooperative, which is the most conservative assumption in terms of secrecy\footnote{This cooperation is possible in e.g. cloud radio access networks (C-RAN), where users' baseband data is centrally stored and processed \cite{Wu-15}-\cite{Peng-16} and hence a malicious user (super-Ev) can exploit it for eavesdropping.}.  The objective is to ensure reliable communications between the Tx and Rx (the reliability criterion) while keeping the Evs ignorant about transmitted information (the secrecy criterion). In an interference-constrained (IC) multi-user environment, such as cognitive radio, the interference generated by the Tx to primary receivers (PR), who represent licensed users of the system, must not exceed certain thresholds.  The secrecy capacity is defined operationally as the largest transmission rate on the Tx-Rx link subject to the reliability and secrecy criteria \cite{Bloch-11}-\cite{Oggier-10}, where the reliability criterion ensures arbitrary low error probability at the Rx while recovering the transmitted message; the secrecy criterion ensures arbitrary low information leakage to the Evs. The Tx has $m$ antennas, while the Rx and each Ev have $n_1$ and $n_{2i}$ ($i=1,2,...,N$) antennas, respectively. In the discrete-time AWGN MIMO channel model, the signals received by the Rx and each Ev can be expressed as
\bal
\label{eq.ch}
\by_1 = \bH_1\bx+\bxi_1, \quad \by_{2i} = \bH_{2i}\bx+\bxi_{2i}
\eal
where $\by_{1(2i)}$ are the respective received signals at Rx ($i$-th Ev), $\bx$ is the transmitted signal, $\bxi_{1(2i)}$ represent zero-mean unit-variance i.i.d. noise at  the Rx ($i$-th Ev) end; $\bH_{1(2i)}$ are the channel matrices collecting channel gains from the Tx to the Rx ($i$-th Ev).  In addition to this and following the interference-constrained model, there are $K$ PRs equipped with $n_{3j}$ ($j=1,2,...,K$) antennas each. The received signal at $j$-th PR  is similarly expressed as
\bal
\label{eq.ch.PR}
\by_{3j} = \bH_{3j}\bx+\bxi_{3j}
\eal
where $\bH_{3j}$ and $\bxi_{3j}$ are the channel matrix and zero-mean unit-variance i.i.d. noise. For future use, let $\bW_k=\bH_k^{+}\bH_k, k=1,2$ and let $\bW_{3j}=\bH_{3j}^{+}\bH_{3j}, j=1,2,...,K$. We assume that the full channel state information (CSI) is available to the Tx, Rx and each Ev (which is motivated by modern adaptive system design, where channel is estimated at the Rx and send back to the Tx via a feedback link; when Evs are just other users in the system, they also share their CSI with the base station).

\begin{figure}[t]%[htbp]
	\centerline{\includegraphics[width=3.5in]{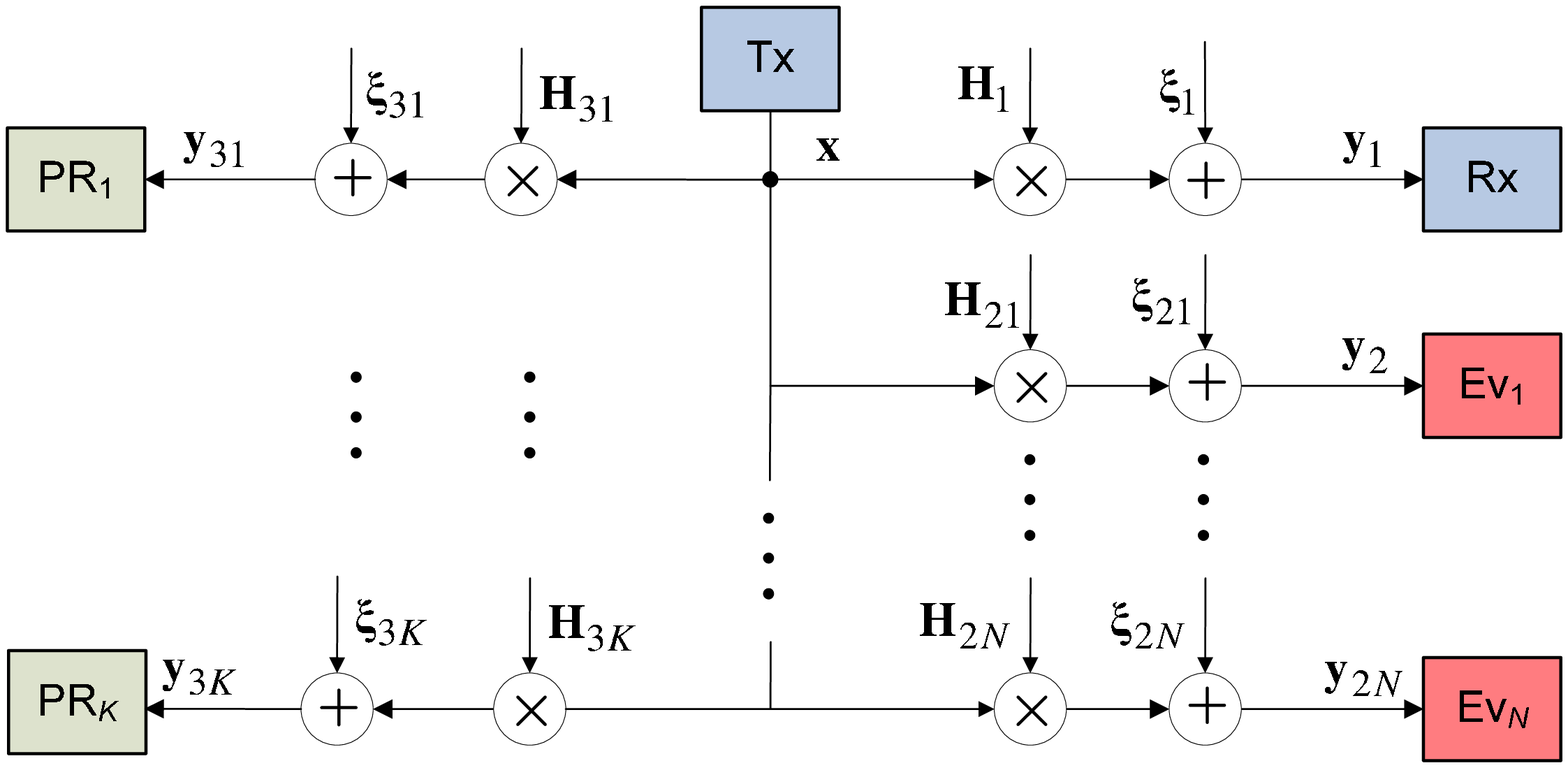}}
	\caption{A block diagram of the Gaussian multi-user MIMO wiretap channel under interference constraints. $\bH_1$, $\bH_{2i}$ and $\bH_{3j}$ are the channel matrices to the Rx, each Ev and PR respectively; $\bx$ is the Tx signal; $\by_1$, $\by_{2i}$ and $\by_{3j}$ are the received signal at the Rx, each Ev and PR respectively; $\bxi_1$, $\bxi_{2i}$ and $\bxi_{3j}$ are respective noise components.}
	\label{fig.1}
\end{figure}

Overall, the transmission is subject to the TPC and multiple IPCs, so that any Tx covariance matrix $\bR=E \left \{\bx\bx^{+}\right \}$ must be in the following feasible set $S_{\bR}$:
\bal
\label{eq.SR}
S_{\bR}= \{\bR\geq  \bo: tr(\bR)\leq  P_T, tr(\bW_{3j}\bR)\leq P_{Ij}\ \forall j \}
\eal
where $P_T,\ P_{Ij}$ are the maximum allowed transmit and interference powers at the Tx and each PR respectively, termed here the TPC and IPC powers. The IPC
\bal
tr(\bW_{3j}\bR)=tr(\bH_{3j}\bR\bH_{3j}^{+})\leq P_{Ij},\  j=1,2..K,
\eal
ensures that the total interference power at the $j$th PR does not exceed the IPC power $P_{Ij}$ so that this PR's performance is not distorted. This type of interference constraints has been widely adopted in the literature for regular systems (no secrecy) \cite{Zhang-08}\cite{Yang-13}-\cite{Zhang-12} as well as for secrecy systems \cite{Zhang-10}\cite{Li-13}\cite{Pei-10}-\cite{Zhang-18}. In this multi-user environment, the secrecy capacity of the interference-constrained WTC is defined operationally as the largest achievable rate on the Tx-Rx link subject to the secrecy, reliability, transmit and interference  power constraints simultaneously. Note that per-antenna power constraints (as in e.g. \cite{Vu-11}\cite{Loyka-17b}), in addition to or instead of the TPC, can also be accommodated by setting some $\bW_{3j}$ to be diagonal matrices with 0-1 entries.

%\vspace*{-.5\baselineskip}
%========================================================================
\subsection{Secrecy Capacity of Interference-Constrained MIMO WTC}
\label{sec.Sec.Cap}
Since the Evs are cooperative, their received signals can be aggregated into a single vector resulting in a single meta-Ev as follows:
\bal
\label{eq.ch.y2}
\by_{2} = \bH_{2}\bx+\bxi_{2}
\eal
where $\by_{2}$, $\bH_{2}$ are aggregated received signals and channel matrices, respectively, $\by_{2}=[\by_{21}^{+}, \by_{22}^{+},..., \by_{2N}^{+}]^{+}$, $\bH_{2}=[\bH_{21}^{+}, \bH_{22}^{+},..., \bH_{2N}^{+}]^{+}$, $\bxi_{2}=[\bxi_{21}^{+}, \bxi_{22}^{+},..., \bxi_{2N}^{+}]^{+}$.
While the MISO case was considered in \cite{Pei-10}, its approach cannot be extended to the full MIMO case since an optimal covariance is not necessarily of rank-1 and there is no equivalent scalar channel allowing quasi-convex reformulation of the original non-convex problem.

A different approach was adopted in \cite{Dong-18}: it is based on the max-min characterization of the secrecy capacity originally developed in \cite{Khisti-10b}\cite{Oggier-10} under the TPC alone and its further extension to the case of the joint constraints (TPC+IPC). This established the operational secrecy capacity of the interference-constrained MIMO WTC above as follows (Gaussian signaling is still optimal in this setting).
\begin{thm}
	\label{thm.C}
	The  operational secrecy capacity of interference-constrained Gaussian MIMO WTC in \eqref{eq.ch}, \eqref{eq.ch.y2} and \eqref{eq.ch.PR} under the TPC and the IPCs in \eqref{eq.SR} can be expressed as
	\bal
	\label{eq.secrecy}
	C=\underset{\bR\in S_{\bR}}{\max}C(\bR)=\underset{\bR\in S_{\bR}}{\max}\ \underset{\bK\in S_{\bK}}{\min}f(\bR,\bK) \qquad \mathrm{(P1)}
	\eal
	where
	\bal
\label{eq.C(R)}
	&C(\bR)=\ln|\bI +\bW_1\bR|-\ln|\bI+\bW_2\bR|,
	\\
	&f(\bR,\bK)=
	\ln|\bI+\bK^{-1}\bH\bR\bH^{+}|-\ln|\bI+\bW_2\bR|,
	\eal
 and $\bH=[\bH_{1}^{+},\bH_{2}^{+}]^{+}$, $S_{\bK}$ is a set of noise covariance matrices of the form
	\bal
	\label{eq.set.K}
		&S_{\bK}=\left \{\bK:\bK=\begin{bmatrix}
		\bI & \bN \\
		\bN^{+} & \bI
	\end{bmatrix},\bK\geq \bo \right \},
	\eal
where $\bN =E\{\bxi_{1}\bxi_{2}^{+}\}$ is noise cross-covariance.
\end{thm}

We emphasize that Theorem 1 characterizes the \textit{operational} secrecy capacity (the largest achievable secrecy rate) rather than an information capacity defined formally as the difference of two mutual information terms, as sometimes done in the literature (without proving its operational significance). The case of multiple \textit{non-cooperating} Evs corresponds to a compound WTC (see e.g. \cite{Liang-09}\cite{Schaefer-15}) and is much more difficult for analysis; its secrecy capacity is not known under interference constraints in general (it is not even known whether Gaussian signaling is optimal). The capacity above is a lower bound to that of the non-cooperative case, since Evs cooperation, while having no effect on the Rx error probability, cannot decrease the information leakage and hence cannot increase the secrecy capacity. However, if there exists a dominant Ev, as in \cite[Proposition 8]{Dong-18}, then Theorem 1 still holds, with $\bW_2$ being the channel Gram matrix of the dominant Ev.

Theorem 1 provides two equivalent characterizations of the secrecy capacity: as a max problem or as a max-min problem. While the first characterization appears to be easier for numerical optimization and is indeed widely-used in the existing literature \cite{Li-13}-\cite{Cumanan-14}, it makes it virtually impossible to prove convergence to a global optimum since the max problem is not convex (since $C(\bR)$ is not concave, unless the channel is degraded, see e.g. \cite{Khisti-10b}), and hence its KKT conditions are not sufficient for global optimality. Provable convergence to a global optimum in this case is out of reach \cite{Boyd-04}\cite{Murty-87}-\cite{Nesterov-04}.

Here, we adopt a different approach based on the max-min representation of the secrecy capacity in \eqref{eq.secrecy}, denoted below as (P1). While this representation involves 2 conflicting optimizations, it is actually easier for numerical optimization, since both optimizations are convex as  $f(\bR,\bK)$ is concave in $\bR$ for any fixed $\bK$ and is convex in $\bK$ for any fixed $\bR$ (see \cite{Khisti-10b}\cite{Loyka-16b}\cite{Dong-18} for further details) and, hence, the respective KKT conditions for both optimizations are jointly sufficient for \textit{\textit{global}} optimality. This opens up a path to develop an iterative algorithm, based on this representation, with a probable convergence to a global (rather than local) optimum.

We caution the reader that while the optimal values of the max and max-min problems in \eqref{eq.secrecy} are the same, the respective optimal covariances are not necessarily the same, i.e. $\bR^* \neq \bR'$ and $C=C(\bR^*) > C(\bR')$ in those cases (see \eqref{eq.ED3}-\eqref{eq.Ex.C(R')<C(R*)} for an example), where $\bR^*$ and $(\bR',\bK')$ are optimal points of the max and max-min problem respectively,
\bal\notag
%\label{eq.secrecy}
\bR^* &=\arg\underset{\bR\in S_{\bR}}{\max}C(\bR),\\
(\bR',\bK') &=\arg\underset{\bR\in S_{\bR}}{\max}\ \underset{\bK\in S_{\bK}}{\min}f(\bR,\bK)
\eal
In some (singular) cases, the difference can be significant (see Fig. \ref{fig.R*=not.R'} and \ref{fig.R*.A2.rand}). This phenomenon never appears without IPCs (under the TPC alone), as in \cite{Khisti-10b}\cite{Oggier-10}\cite{Loyka-15}, where both problems always share the same optimal covariance matrix, $\bR^* = \bR'$. To address this issue,  Algorithm 2 is developed in Section \ref{sec.algorithm2}, which computes iteratively an optimal covariance matrix $\bR^*$ in these singular cases. Its convergence is also proved. Finally, it should also be noted that $\bR^*$ is not necessarily unique (see e.g. Example 2 in \cite{Loyka-18}), which motivates the power minimization problem (P4) in \eqref{eq.P4}.

%\vspace*{-1\baselineskip}
%\newpage
%======================================================================
\section{Capacity-Achieving Signaling Under Interference Constraints}
\label{sec.algorithm}
In this section, we propose an iterative algorithm to solve (P1) numerically and prove its convergence to a global optimum. Performing separately $\max$ and $\min$ optimizations in the max-min part of \eqref{eq.secrecy} immediately faces a serious and fundamental difficulty of achieving or proving convergence of the algorithm due to its oscillatory behaviour, which is due to conflicting (max-min) optimization operations. To overcome this difficulty, we use the residual form of Newton method where both optimizations (max and min) are done simultaneously, so that the residual of the KKT conditions is reduced at each iteration and it converges monotonically to zero as the algorithm progresses (see e.g. \cite{Boyd-04} for more details on this general approach). This opens up a way to a provable convergence to a global optimum, which is out of reach for the max problem in \eqref{eq.secrecy} due to its non-convex nature.

We develop below an iterative algorithm, which is able to handle any number of interference constraints and which does not require advance knowledge of which constraint is active and which is not.  This algorithm is based on the max-min representation of the secrecy capacity in \eqref{eq.secrecy} and includes the barrier method, the residual-form Newton method and the backtracking line search, see e.g. \cite{Boyd-04} for more details on these algorithms. Unlike generic convex optimization algorithms or solvers, our algorithm here is specifically tailored for secrecy rate maximization in multi-user Gaussian MIMO WTC under interference constraints. Its convergence to a global optimum is rigorously proved, even when the WTC is not degraded and hence the max problem in \eqref{eq.secrecy} is not convex. This is a distinct advantage not found in other known algorithms, e.g. in \cite{Li-13}-\cite{Cumanan-14}, where either no convergence at all is proved or where only convergence to a stationary point is proved, which is not necessarily a global maximum, as discussed above.

The key idea of the barrier method is to substitute the original objective function $f(\bR,\bK)$ by a modified one $f_{t}(\bR,\bK)$, which includes additional barrier terms as follows:
\bal
\label{eq.barrier.ft}
\notag
f_{t}(\bR,\bK)=f(\bR,\bK) &+I_{1}(\bR)+I_{2}(\bR)\\
&+\sum_{j} I_{3j}(\bR)-I_{4}(\bK)
\eal
where $t>0$ is the barrier parameter and
\bal
\label{eq.barrier.I1}
&I_{1}(\bR)=t^{-1}\ln|\bR|,\\
\label{eq.barrier.I2}
 &I_{2}(\bR)=t^{-1}\ln(P_{T}-tr(\bR)),\\
 \label{eq.barrier.I3}
&I_{3j}(\bR)=t^{-1}\ln(P_{Ij}-tr(\bW_{3j}\bR)),\\
\label{eq.barrier.I4}
 &I_{4}(\bK)=t^{-1}\ln|\bK|.
\eal
so that all inequality constraints are absorbed in the respective barrier terms $I_1 - I_4$. Note that the domain of $f_{t}(\bR,\bK)$ is  $\bR \in S'_{\bR}, \bK \in S'_{\bK}$ where
\bal\notag
&S'_{\bR}= \{\bR\in S_{\bR}: \bR> \bo,\  tr(\bR) < P_T,\ tr(\bW_{3j}) < P_{Ij}\},\\
&S'_{\bK}= \{\bK\in S_{\bK}: \bK> \bo \},
\eal
i.e. $\bR, \ \bK$ are strictly inside of the original feasible sets $S_{\bR},\ S_{\bK}$ (but may approach the boundary arbitrary closely - this is a key feature of the barrier method). Note also that $f_{t}(\bR,\bK)$ is convex-concave in the right way, i.e. concave in $\bR$ for any fixed $\bK$ and convex in $\bK$ for any fixed $\bR$, so that the respective optimization problems are convex and their KKT conditions are jointly sufficient for global optimality.

In the proposed algorithm, we use the residual-form Newton method to compute an  optimal point $\{\bR(t),\bK(t)\}$ of (P2) below for a fixed $t$ in an iterative way and with high accuracy. To facilitate implementation, we use real rather than complex variables. To reduce the number of variables and improve the efficiency, we exploit the symmetry of $\bR$ and $\bK$ and use $\bx=veh(\bR)$ and $\by=vec(\bN)$ as independent variables to represent $\bR$ and $\bK$, where $vec(\bN)$ operator stacks all columns of $\bN$ on top of each other and $veh(\bR)$ does so for the lower-triangular part of $\bR$. Since $vec(\bN)$ is used as independent variables to represent $\bK$, the equality constraint in \eqref{eq.set.K} is satisfied automatically. The original max-min problem (P1) in \eqref{eq.secrecy} is transformed into the following unconstraint problem:
\bal
\label{eq.barrier}
\textrm{(P2)}\ \ \underset{\bx}{\max}\ \underset{\by}{\min}\ f_{t}(\bR,\bK)
\eal
so that its KKT conditions are simply the stationarity conditions:
\bal
\label{eq.norm}
\br(\bz)=\nabla_{\bz}f_t=\bo
\eal
where
\bal
\label{eq.r(z)}
\bz=\begin{bmatrix}
	\bx	\\ \by
	\end{bmatrix},\
 \br(\bz)=\begin{bmatrix}
  \nabla_{\bx}f_t	\\ \nabla_{\by}f_t
 \end{bmatrix}
\eal
are the aggregate vector of the variables and the residuals respectively. In the residual-form Newton method, the optimality condition $\br(\bz)= 0$ is iteratively solved using 1st-order  approximation of $\br(\bz)$ at each step (which corresponds to the second-order approximation of the objective):
 \bal
 \label{eq.norm.taylor}
 \br(\bz_k+\Delta \bz)=\br(\bz_k)+D\br \Delta \bz+o(\Delta \bz)=0.
 \eal
where $\bz_k$ and $\Delta \bz$ are the current variables and their updates respectively at iteration $k$, and where $D\br$ is the derivative of $\br(\bz)$, i.e. the Hessian  of $f_t(\bx,\by)$:
\bal
\label{eq.KKT.hessian}
D\br=\begin{bmatrix}
	\nabla_{\bx\bx}^{2}f_t & \nabla_{\bx\by}^{2}f_t\\
\nabla_{\by\bx}^{2}f_t & \nabla_{\by\by}^{2}f_t
\end{bmatrix}.
\eal
Closed-form expressions for gradients and Hessians are given in the Appendix.  By ignoring $o(\Delta \bz)$, \eqref{eq.norm.taylor} can be reduced to a system of linear equations in $\Delta\bz$:
\bal
\label{eq.update}
\br(\bz_k)+D\br \Delta \bz=0
\eal
which can be solved numerically using any of the existing (and efficient) techniques. When Hessian $D\br$ is non-singular, $\Delta \bz$ in \eqref{eq.update} has a unique solution. In our case, the non-singularity of $D\br$ at each step of the Newton method is rigorously established below. After computing $\Delta \bz$ from \eqref{eq.update}, $\bz$ is updated as follows
\bal
\bz_{k+1}=\bz_k+s\Delta \bz
\eal
where $k$ denotes the Newton iteration number (step) and where $s>0$ is the step size, which can be found via backtracking line search (see e.g. \cite{Boyd-04} for a background on this method). The Newton method in combination with the backtracking line search is guaranteed to reduce the residual
norm $|\br(\bz)|$ at each step, which follows from the  respective norm-reduction property \cite{Boyd-04},
so that for sufficiently small $s$, the residual norm shrinks at each iteration approaching $\br(\bz_k)=0$ as $k$ increases. After several iterations, the convergence becomes quadratic (see \cite{Boyd-04} for related definitions and analysis) and hence very fast, so that the optimal point $(\bR(t),\bK(t))$ of the problem \eqref{eq.barrier} can be approached with any desired accuracy in a small to moderate number of steps. Following the barrier method, the problem in \eqref{eq.barrier} is solved for sequentially increasing $t$, where the optimal point of the previous $t$ serves as an initial point for the new, increased $t$, thus minimizing the total number of Newton iterations required \cite{Boyd-04}. It can be shown that $f(\bR(t),\bK(t))\rightarrow C$ as $t \rightarrow \infty$ so that any desired accuracy can be reached (see Proposition \ref{prop.gap.C} below).

The proposed algorithm is shown below, where $\alpha$ is the percentage of the linear decrease in the residual norm one is willing to accept in the backtracking line search; $\beta$ and $\eta$ are the parameters controlling reduction in step size $s$ and increase in barrier parameter $t$ at each iteration of the respective loop of the algorithm, $\epsilon$ is the target residual accuracy, $t_0$ and $t_{max}$ are initial and maximum values of the barrier parameter; $t$ varies from $t_0$ to $t_{max}$, where the latter controls the accuracy of the barrier method so that the inaccuracy in the secrecy capacity due to the barrier method does not exceed $\max\{m+1+K,n_1+n_2\}/t_{max}$. $\bz_{0}=[\bx_0^T,\by_0^T]^T$ is an initial point defined as follows
\bal\notag
\label{eq.initial.a}
&\bx_{0}=veh(P_{T}\bI/a),\ \by_{0}=\bo\\
&a = 2 \max \{ m, \{ tr(\bW_{3j})P_T /P_{Ij}\} \},
\eal
so that $\bR_0 \in S'_{\bR}, \bK_0 \in S'_{\bK}$. Note that $\bR_0$ represents isotropic signaling satisfying all power constraints and $\bK_0$ represents uncorrelated noise. Numerical experiments show that this initial point results in fast convergence in all studied cases. While the barrier method generates a sequence of $\{\bR(t),\bK(t)\}$ which are strictly inside the feasible set (e.g. non-singular), they may approach the boundary arbitrary closely, thus representing a rank-deficient solution. In this case, non-zero but very small eigenvalues of $\bR(t)$ can be rounded off to zero facilitating low-complexity (low-rank) implementation, which includes beamforming as a special case.

%\newpage
\begin{algorithm}[h]
	\caption{(for optimal secure signaling under interference constraints)}
	\label{alg.1}
	\begin{algorithmic}
		\State \bRequire\  $\bz_{0}$, $0<\alpha <0.5$, $0<\beta<1$, $t_{max}> t_{0}>0$, $\eta>1$, $\epsilon >0 $.
        \State 1. Set $t=t_0$; compute $\br(\bz_0)$ via \eqref{eq.r(z)}.
		\Repeat\ \ (barrier method)
        \State 2. Set $k=0$.
		\Repeat\ \ (Newton method)
		\State 3. Compute the Hessian matrix $D\br$ via \eqref{eq.KKT.hessian}.
		\State 4. Compute update $\Delta \bz$ via \eqref{eq.update}.
		\State 5. Set $s=1$.
		\Repeat\ \ (backtracking line search)
		\State 6. $s:=\beta s$.
        \State 7. Update $\bz_{k+1}=\bz_{k}+s\Delta \bz$; compute $\br(\bz_{k+1})$
		\Until {$|\br(\bz_{k+1})|\leqslant (1-\alpha s)|\br(\bz_{k})|$ and  $\bR_{k+1} \in S_{\bR}',\bK_{k+1} \in S_{\bK}'$}
		\State 8. $k:=k+1$.
		\Until{$|\br(\bz_{k})|\leqslant \epsilon$}
		\State 9. Compute $f(\bR_k,\bK_k)$, $C(\bR_k)$.
        \State 10. Set $\bz_{0}:=\bz_{k}$ as a new starting point.
        \State 11. Update $t:=\eta t$.
		\Until{$t>  t_{max}$}
        \State 12. Output: $(\bR_k,\bK_k)$, $f(\bR_k,\bK_k)$, $C(\bR_k)$.
	\end{algorithmic}
\end{algorithm}

%\vspace*{-1\baselineskip}
%======================================================================
\subsection{Analysis of Algorithm 1}
\label{sec.A1a}

In this section, we prove the convergence of Algorithm 1 to a globally-optimal solution of the problem in \eqref{eq.secrecy} using the steps of the convergence analysis in \cite{Boyd-04} and adapting them properly to the current setting. First, from the residual norm-reduction property of the Newton method (see Sec. 10.3 in \cite{Boyd-04}),
\bal
\label{eq.ds.r}
\frac{d}{ds} |\br(\bz_k+s\Delta \bz)| = -|\br(\bz_k)| \le 0
\eal
so that the termination condition of the backtracking line search in Algorithm 1 is satisfied for a sufficiently-small $s>0$,
\bal
%\label{eq.ds.r}
|\br(\bz_k+s\Delta \bz)| = (1-s)|\br(\bz_k)|  +o(s) \le (1-\alpha s)|\br(\bz_k)|
\eal
where $0<\alpha <0.5$, and hence
\bal
\label{eq.r(zk)}
|\br(\bz_{k+1})| \le |\br(\bz_k)|
\eal
so that $\{|\br(\bz_k)|\}$ is a decreasing sequence that converges (since  it is bounded from below by 0); from \eqref{eq.ds.r}, a convergence point is 0 (otherwise, $|\br(\bz)|$ could be further reduced as the inequality in \eqref{eq.r(zk)} is strict if $|\br(\bz_k)|>0$), i.e. to a point that solves the KKT conditions. This point is globally-optimal since the KKT conditions are sufficient for global optimality of (P2) due to the convex-concave nature of $f_t(\bR,\bK)$, as explained above.

It remains to show that, (i) at each step of the Newton method, \eqref{eq.update} can be solved to obtain update $\Delta\bz$, and that (ii) $f_t(\bR(t),\bK(t))$ will approach $f(\bR',\bK') = C$ arbitrary closely as $t$ increases, where $(\bR',\bK')$ is an optimal (saddle) point of (P1).

To establish first point, it is sufficient to show that the Hessian $D\br$ is non-singular at each Newton step

\begin{prop}
\label{prop.non-singular.KKT}
Consider the max-min problem $\mathrm{(P2)}$ in \eqref{eq.barrier}. Its Hessian $D\br$ as defined in \eqref{eq.KKT.hessian} is non-singular for each  $t>0$, $\bR \in S'_{\bR},\bK \in S'_{\bK}$.
\end{prop}
\begin{proof}
See Appendix.
\end{proof}

In fact, Proposition \ref{prop.non-singular.KKT} ensures that the update equation \eqref{eq.update} has a \textit{unique} solution at each Newton step. To demonstrate second point, we give below a sub-optimality bound for the barrier method, from which it follows that $f(\bR(t),\bK(t)) \rightarrow C$ as $t \rightarrow \infty$.

\begin{prop}
\label{prop.gap.C}
	For each $t>0$, the gap of the barrier method used in \eqref{eq.barrier} can be upper bounded as follows:
	\bal
	\label{eq.opt.gap.2}
	|f(\bR(t),\bK(t)) - C| \le \max(m_R,n_K)/t
	\eal
where $\bR(t),\bK(t)$ are the optimal signal and noise covariance matrices returned by the barrier method for a given $t$; $n_K = n_1+n_2$, $m_R = m+ 1 + K$ and $K$ is the number of IPCs.
\end{prop}
\begin{proof}
To establish the bound, consider first the min part of (P1) in \eqref{eq.secrecy} for a fixed $\bR=\bR(t)>0$ and use the analysis of the barrier method in Sec. 11.6 of \cite{Boyd-04} to obtain an upper bound:
\bal\notag
f(\bR(t),\bK(t)) &\le \min_{\bK\in \sS_{\bK}} f(\bR(t),\bK)+n_K/t \\
&\le \max_{\bR\in \sS_{\bR}} \min_{\bK\in \sS_{\bK}} f(\bR,\bK)+n_K/t
\eal
where $n_K = n_1+n_2$ accounts for the constraint $\bK\ge 0$. Consider now the max part of (P1) for a fixed $\bK=\bK(t)>0$ and use the same approach to obtain
\bal\notag
f(\bR(t),\bK(t)) &\ge \max_{\bR\in \sS_{\bR}} f(\bR,\bK(t)) -m_R/t\\
 &\ge \min_{\bK\in \sS_{\bK}} \max_{\bR\in \sS_{\bR}} f(\bR,\bK) -m_R/t
\eal
where $m_R = m+ 1 + K$; $m$ accounts for the positive semi-definite constraint $\bR\ge 0$, while $1$ and $K$ account for the TPC and IPC respectively. Combining these two bounds, one obtains
\bal
C -m_R/t \le f(\bR(t),\bK(t)) \le C + n_K/t
\eal
from which \eqref{eq.opt.gap.2} follows, where we have used the saddle-point property of the problem (P1):
\bal
C = \min_{\bK\in \sS_{\bK}} \max_{\bR\in \sS_{\bR}} f(\bR,\bK) = \max_{\bR\in \sS_{\bR}} \min_{\bK\in \sS_{\bK}} f(\bR,\bK)
\eal
which follows from Von Neumann mini-max theorem \cite{Dong-18}.
\end{proof}

Note that the original objective $f(\bR,\bK)$ is used in \eqref{eq.opt.gap.2}, not the modified one $f_t$. The bound in \eqref{eq.opt.gap.2} can be used in practice to set up $t_{max}$ to meet a target accuracy $\Delta C$ in terms of the achieved secrecy rate $f(\bR(t),\bK(t))$: if $|C - f(\bR(t),\bK(t))| \le \Delta C$ is needed, then setting
\bal
t_{max} \ge \max(m_R,n_K)/\Delta C
\eal
will satisfy this requirement.

\section{Optimal Covariance in the Singular Case}
\label{sec.algorithm2}

Algorithm 1 can be used to evaluate the secrecy capacity in the general case by evaluating numerically $f(\bR',\bK')=C$ at saddle point $(\bR',\bK')$ (i.e. the optimal point of the max-min problem (P1) in \eqref{eq.secrecy}). However, $\bR'$ is not necessarily an optimizer of $C(\bR)$, i.e. not an optimal Tx covariance $\bR^* = \argmax_{\bR\in S_{\bR}} C(\bR)$, so that $C(\bR') < C(\bR^*)=f(\bR',\bK')$ is possible. This happens when the TPC is inactive and, for all active IPCs, the sum $\sum_{j+} \bW_{3j}$ is singular (i.e. the intersection of their null spaces is not empty). The following example \cite{Dong-18} illustrates this point. Let
\bal
\label{eq.ED3}
\bH_1=diag\{1,0\},\ \bH_2=diag\{0,1\},\ \bW_3=diag\{1,0\}
\eal
It is straightforward to see that the saddle-point $(\bR',\bK')$  is
\bal
\label{eq.Ex.K'R'}
\bK'=\bI,\ \bR'=diag\{\min(P_T,P_I), a\}
\eal
where $a$ is any in the interval $0\le a \le (P_T-P_I)_+$, so that $\bR'$ is not unique if $P_T> P_I$. The optimal covariance is
\bal
\bR^*=diag\{\min(P_T,P_I), 0\}
\eal
Thus, $\bR^* \neq \bR'$ (unless $a=0$) and
\bal
C = f(\bR',\bK') = \ln(1+\min(P_T,P_I))
\eal
for any $a$. However, if one sets $a = (P_T-P_I)_+$, then
\bal
\label{eq.Ex.C(R')<C(R*)}
C(\bR') = \ln\frac{1+\min(P_T,P_I)}{1+(P_T-P_I)_+} < C= C(\bR^*)
\eal
where the inequality holds if $P_T>P_I$ (negative $C(\bR')$ is interpreted as zero rate). Hence, $\bR'$ is not an optimal transmit covariance $\bR^*$ (one maximizing the secrecy rate $C(\bR)$). We conclude that while the application of Algorithm 1 is possible to find the secrecy capacity via $C= f(\bR',\bK')$, it cannot be used to find $\bR^*$ in the singular case, since $\bR' \neq \bR^*$ is possible and, furthermore, $C(\bR') < C= C(\bR^*)$ is also possible (as a side remark, we note that this effect disappears if the IPCs are removed, since the TPC is always active in this case).

Therefore, the singular case needs special treatment to establish an optimal signaling strategy (optimal covariance), not just the capacity. This is done below via Algorithm 2, which incorporates Algorithm 1 and bisection search to find an optimal covariance as well as the least Tx power required to achieve the secrecy capacity (this may be smaller than the TPC power $P_T$ in the singular case).

To this end, let $C(P_T)$ be the secrecy capacity as a function of TPC power $P_T$, with all interference constraint powers being fixed, and let
\bal
\label{eq.P0}
P_0= \min\{P: C(P_T)\le C(P)\ \forall P_T \ge 0\}
\eal
so that $C(P_T) \le C(P_0) \ \forall P_T \ge 0$, i.e. $C(P_T)$ saturates at $C(P_0)$ as $P_T$ increases; $P_0 = \infty$ corresponds to no saturation. It follows from the definition of $P_0$ that $P_{T,min}=\min\{P_T, P_0\}$ is the minimum Tx power required to achieve the capacity $C(P_T)$. Note that $P_{T,min} < P_T$ if $P_T > P_0$, i.e. Tx power saving is possible and hence it is important to evaluate $P_0$ as well.

The following general properties of the function $C(P_T)$ are needed below to construct an algorithm and to prove its convergence. To the best of our knowledge, these properties of the secrecy capacity never appeared in the literature before, even without interference constraints. We will assume below that $C(P_0) > 0$, i.e. $C(P_T)$ is not identically 0 for all $P_T$ (which would be the case for a reversely-degraded channel).

\begin{prop}
\label{prop.Properties}
Let $C(P_0) > 0$. The secrecy capacity $C(P_T)$ as a function of TPC power $P_T$ (under fixed $P_{Ij}$) has the following properties:

1. $C(P_T)$ is a non-decreasing function of $P_T$; strictly-increasing for any  $P_T < P_0$.

2. $C(P_T)$ is a concave, continuous function of $P_T$.

3. If $C(P_T) = C(P_1)$ for some $P_1 > P_T$, then this holds for any $P_1 > P_T$. Equivalently, if $C(P_T)'_+=0$, then $C(P_1)'_+=0$ for any $P_1 > P_T$, where $C(P_T)'_+$ is the right derivative; additionally, $C(P_T)'_+=0$ for any $P_T \ge P_0$.

4. If $C(P_1) < C(P_T)$ for some $P_1$, then $C(P_2) < C(P_1)$ for any $P_2 < P_1$. Equivalently, if $C(P_T)'_- > 0$, then $C(P_1)'_- > 0$ for any $P_1 < P_T$, , where $C(P_T)'_-$ is the left derivative; additionally, $C(P_T)'_- >0$ for any $P_T \le P_0$, i.e. $C(P_T)$ is strictly increasing for any $P_T < P_0$.
\end{prop}
\begin{proof}
See Appendix
\end{proof}

Thus, $C(P_T)$ is concave, non-decreasing, and strictly-increasing for $P_T < P_0$. The rate of increase slows down with $P_T$. Note that for $P_T > P_0$, the capacity $C(P_T)$ can be achieved with smaller Tx power $P_0$. %It is straightforward to show that the secrecy capacity $C(P_{Ij})$ as a function of interference power $P_{Ij}$ has the same properties as $C(P_T)$ above, with $P_{0j}$ in place of $P_{0}$.
We will need below the following result to deal with the singular case.

\begin{prop}
\label{prop.muPT}
Let $\mu(P_T)$ be a Lagrange multiplier, as a function of TPC power $P_T$, responsible for the TPC in (P1). Then, $\mu(P_T)>0$ for any $P_T < P_0$, i.e. the TPC is always active below $P_0$, and $\mu(P_T)=0$ for any $P_T > P_0$.
\end{prop}
\begin{proof}
See Appendix.
\end{proof}

%As a side remark, we note that Proposition \ref{prop.muPT} says nothing about $\mu(P_T)$ at $P_T=P_0$. This, however, is not necessary to construct Algorithm 2 below.

Based on this Proposition, we are now able to construct an iterative algorithm to evaluate optimal covariance in the singular case numerically with any desired accuracy. The key idea for the $P_T \ge P_0$ case (which is necessary for singularity) is to identify the saturation point $P_0$ and to apply Algorithm 1 with TPC power $P$ slightly less than $P_0$ (so that $\mu(P)>0$ and hence the TPC is active thus avoiding the singularity in this way), which achieves the secrecy rate arbitrary close to the capacity $C(P_T)=C(P_0)$ as $P$ approaches $P_0$ from below, and gives a covariance matrix achieving this secrecy rate as well.

\begin{algorithm}[h]
	\caption{(optimal signaling in the singular case)}
	\begin{algorithmic}
		\State \bRequire\  $\delta$, $\epsilon$
		\State 1. Compute $C = f(\bR',\bK')$ using Algorithm 1 for a given  $P_T$.
       \State 2. Set $P_{min}=0,\ P_{max}=P_T,\ P= P_T/2$.
       \State 3. Compute $f(\bR',\bK')$ under new TPC $tr(\bR)\leq P$ using Algorithm 1.
        \Repeat\ \ (bisection search)
		\State 4. If  $f(\bR',\bK') < (1-\epsilon)C$, set $P_{min}=P$; otherwise, set $P_{max}=P$.
        \State 5. Set $P=(P_{min}+P_{max})/2$.
		\State 6. Compute $f(\bR',\bK')$ under TPC $tr(\bR)\leq P$ using Algorithm 1.
		\Until{$P_{max}-P_{min} \leqslant \delta P_T$}
        \State 7. Compute $f(\bR',\bK')$ under TPC $tr(\bR)\leq P_{min}$ using Algorithm 1.
        \State 8. Output $\bR',\ C(\bR'),\ \Delta C = C - C(\bR'),\ P_{T,min}=P_{min}$, $\Delta P = P_{max} - P_{min}$.
    \end{algorithmic}
\end{algorithm}

Algorithm 2 returns nearly-optimal covariance $\bR'$ as well as its achieved secrecy rate $C(\bR')$ and its distance $\Delta C$ to the secrecy capacity $C$. In addition, the algorithm returns an approximate value of $P_{T,min}=\min\{P_T, P_0\}$, i.e. the minimum Tx power required to achieve $C(P_T)$, as well as its accuracy $\Delta P$.  Note that $\Delta C$ and $\Delta P$ can be made as small as necessary by setting sufficiently small $\delta$ and $\epsilon$ (this follows from the continuity of all functions involved as well as the compactness of the feasible set for any finite $P_T$, in addition to the nature of the bisection).

The condition in Line 4 is set to account for numerical imprecision effects in computing $f(\bR',\bK')$. While in theory one can set $\epsilon =0$, this can result in numerical instability in practice in some cases. Typical values of $\epsilon$ range between $10^{-2}$ (1\% accuracy) to $10^{-6}$; $\delta$ controls the accuracy of computed $P_{T,min}$ and $\delta = 10^{-2}$ corresponds to 1\% accuracy with respect to $P_T$.

\subsection{Analysis of Algorithm 2}

Here we provide a convergence analysis of Algorithm 2 to justify the claims above. To simplify the discussion, we consider first the case of $\epsilon =0$ and neglect the numerical imprecision effects (in particular, the imprecision of Algorithm 1, whose accuracy can be very high even for a small number of Newton steps), which is a standard assumption in the literature (see e.g. convergence analysis in \cite{Boyd-04}). The convergence of sufficiently small but non-zero $\epsilon >0$ will follow from the continuity of all functions involved.

Let $P_{min,k}$, $P_{max,k}$ and $P_k$ be the power values set in Line 4 and 5 of Algorithm 2, i.e. at $k$-th iteration of the bisection. Note that, due to the nature of the bisection, $\Delta_k = P_{max,k} - P_{min,k}$ is reduced by a factor of 2 at each step, so that
\bal
\Delta_k = P_{T}/2^k
\eal
The following proposition gives further important properties.
\begin{prop}
\label{prop.A2.prop}
The following holds at $k$-th iteration of the bisection in Algorithm 2 with $\epsilon=0$:
\bal
\label{eq.A2.prop.1}
P_{min,k} < P_k < P_{max,k} \le P_T
\eal
and $\{P_{min,k}\}$, $\{P_{max,k}\}$ are monotonically increasing and decreasing sequences, respectively. If $P_T \ge P_0$, then
\bal
\label{eq.A2.prop.2}
P_{min,k} < P_0 \le P_{max,k} \le P_T
\eal
If $P_T < P_0$, then
\bal
\label{eq.A2.prop.3}
P_{min,k} = P_T(1 - 2^{-k}),\ P_{max,k}=P_T
\eal
\end{prop}
\begin{proof}
To prove first two inequalities in \eqref{eq.A2.prop.1}, use $P_k =(P_{min,k}+ P_{max,k})/2$ and $\Delta_k > 0$ for any $k$. The last inequality is by construction of the algorithm, i.e. from $P_{max,0}=P_T$ and $P_{max,k+1} \le P_{max,k}$, which follows from the fact that either $P_{max,k+1} = P_{max,k}$ or $P_{max,k+1} =P_k < P_{max,k}$. Likewise, $P_{min,k+1} \ge P_{min,k}$, since either $P_{min,k+1} = P_{min,k}$ or $P_{min,k+1} =P_k > P_{min,k}$.

First inequality in \eqref{eq.A2.prop.2} follows from Line 4 (with $\epsilon=0$), which implies that $P_{min,k} = P$ iff $f(\bR',\bK') = C(P) < C = C(P_0)$ so that, from the monotonically-increasing property of $C(P)$ in Proposition \ref{prop.Properties} and the initial condition $P_{min,0}=0$, $P_{min,k} < P_0$. second inequality in \eqref{eq.A2.prop.2} is established in a similar way.

If $P_T < P_0$, then $P_{max,k}=P_T$, since $f(\bR',\bK') = C(P) = C = C(P_T)$ implies $P= P_T$, from the monotonically-increasing property of $C(P)$ in Proposition \ref{prop.Properties}, and the initial condition is $P_{max,0}=P_T$. First equality in \eqref{eq.A2.prop.3} follows from second and $\Delta_k = P_{T}/2^k$.
\end{proof}

Since $\Delta_k \rightarrow 0$ as $k \rightarrow \infty$, it follows from Proposition \ref{prop.A2.prop} that
\bal
P_{min,k},\ P_{max,k},\ P_k \rightarrow \min\{P_T, P_0\} = P_{T,min}
\eal
so that the minimum required power $P_{T,min}$ can be evaluated with any desired accuracy. Furthermore, the inaccuracy does not exceed $\Delta_k$ and, since $\Delta_k = P_{T}/2^k$, the convergence is exponentially fast, so that very few steps are required in practice to achieve high accuracy. The number $k_{\delta}$ of steps needed to achieve the target accuracy $\delta P_T$ is, from $\Delta_k \le \delta P_T$,
\bal
k_{\delta} = \left\lceil\log_2 \frac{1}{\delta}\right\rceil
\eal

Further note that Line 7 of Algorithm 2 evaluates $\bR'$ under $tr(\bR)\leq P_{min, k_{\delta}} < P_0$, where $k_{\delta}$ is the total number of bisections, so that $\mu(P_{min, k_{\delta}}) >0$ under this condition (since, from Proposition \ref{prop.muPT}, $\mu(P) > 0$ if $P < P_0$) and hence $ f(\bR',\bK')= C(\bR') = C(P_{min,k_{\delta}})$, i.e. $\bR'$ is a maximizer of $C(\bR)$ as well under the TPC power $P_{min, k_{\delta}}$. From the continuity of $C(P)$, $C(P_{min, k_{\delta}})\rightarrow C(P_T)$ as $k_{\delta} \rightarrow\infty$, or equivalently, $\Delta C \rightarrow 0$ as $\delta \rightarrow 0$, i.e. arbitrary high accuracy can be achieved in terms of the secrecy rate as well, with exponentially-fast convergence.

The case of non-zero $\epsilon >0$ can be considered in a similar albeit more technical way. Let $C^{-1}(\cdot)$ be the inverse function of $C(P)$ and
\bal
P_{0\epsilon} = C^{-1}((1-\epsilon)C(P_0))
\eal
so that $C(P_{0\epsilon}) = (1-\epsilon)C(P_0)$. The same steps as in the proof of Proposition \ref{prop.A2.prop} can be used to establish \eqref{eq.A2.prop.1}-\eqref{eq.A2.prop.2} with $P_{0\epsilon}$ in place of $P_0$. Eq. \eqref{eq.A2.prop.3} applies as long as $P_{min,k} < P_{T\epsilon} = C^{-1}((1-\epsilon)C(P_T))$, after which \eqref{eq.A2.prop.2} applies with $P_{T\epsilon}$ in place of $P_0$. Note that $P_{0\epsilon} < P_0$,  $P_{T\epsilon} < P_T$, and $P_{0\epsilon} \rightarrow P_0$, $P_{T\epsilon} \rightarrow P_T$ as $\epsilon \rightarrow 0$, so that similar accuracy and performance is expected for sufficiently small but non-zero $\epsilon$.

%============================================================================
\section{Dual problems}
\label{sec.dual}
Motivated by the energy efficiency issues (green communications, battery life etc.), one is lead to consider the following problem dual of \eqref{eq.secrecy}, which is to minimize \textit{globally} the total Tx power subject to the secrecy and interference power constraints:
\bal%\notag
\label{eq.P3}
&\mathrm{(P3)}\qquad \min_{(P,\bR)} P\ \ \mbox{s.t.}\ \ (P,\bR)\in S_3
\eal
where the feasible set $S_3$ is
\bal\notag
\label{eq.S3}
S_3 = \{(P,\bR):\ C(\bR)&\ge C_0,\ \bR\ge 0,\ tr(\bR)\le P,\\
 &tr(\bW_{3j}\bR)\le P_{Ij}\}
\eal
and $C_0$ is the target secrecy rate and $C(\bR)\ge C_0$ is the secrecy constraint. Note that this problem is not convex in general, since $C(\bR)$ is not concave (unless the channel is degraded), and, hence, powerful tools of convex optimization cannot be used to solve it numerically (i.e. to find a global optimum).

In addition to this problem, since an optimal covariance $\bR^*$ of the max problem in \eqref{eq.secrecy} is not necessarily unique (see e.g. Example 2 in \cite{Loyka-18}), a new problem emerges: among all optimal covariances $\bR^*$, find one with the least trace (i.e. the minimum Tx power):
\bal%\notag
\label{eq.P4}
&\mathrm{(P4)}\qquad \min_{\bR\in S_1^*} tr(\bR)
\eal
where $S_1^*$ is the set of all optimal covariances $\bR^*$ of the max problem in \eqref{eq.secrecy}. Since an explicit characterization of this set is not known in the general case (it is not even known whether this set is convex), standard optimization tools (including convex optimization) seem to be inapplicable, making this problem difficult for a direct attack.

The following Proposition shows that these problems have identical solutions and that Algorithm 2 solves both of them.

\begin{prop}
\label{prop.dual}
Let $P_3$ and $P_4$ be the optimal values of (P3) and (P4), and let $C_0$ be the optimal value of (P1), i.e. $C_0=C(P_T)$. Then,
\bal\notag
\label{eq.P3=P4}
&P_3=P_4=\min\{P_T,P_0\},\\
&\bR^*_3=\bR^*_4,\ C(\bR^*_3)=C(\bR^*_4)=C(P_T)
\eal
where $\bR^*_3, \bR^*_4$ are optimal covariances of (P3) and (P4). If $P_T\ge P_0$, then $P_3=P_4=P_0$ and Algorithm 2 also solves (P3) and (P4). If $P_T< P_0$, then $\bR^*_3=\bR^*_4=\bR^*=\bR'$ and $P_3=P_4=P_T$ (Algorithm 1 is sufficient, Algorithm 2 is not necessary in this case).
\end{prop}
\begin{proof}
First, we show that $P_3=P_4$ if $C_0=C(P_T)$. Indeed, it follows from \eqref{eq.S3} that $(P_T,\bR^*)\in S_3$. Hence,
\bal
P_3\le P_T,\ tr(\bR_3^*)\le P_3\le P_T,\ C(\bR_3^*)\ge C(P_T)
\eal
which implies that $\bR_3^* \in S_{\bR}$ so that $C(\bR_3^*)\le C(P_T)$. Therefore, $C(\bR_3^*)= C(P_T)$, i.e $\bR_3^*$ is also optimal for the max problem in \eqref{eq.secrecy}: $\bR_3^*\in S_1^*$. This implies $tr(\bR_3^*)=P_3 \ge P_4$. To show the opposite inequality, note that
\bal
C(\bR^*_4)=C(P_T)=C_0,\ tr(\bR^*_4)=P_4
\eal
and hence $(P_4,\bR^*_4) \in S_3$ so that $P_3\le P_4$. Therefore, $P_3=P_4$, as required. This also implies that $\bR^*_4$ is optimal for (P3) (since $tr(\bR_4^*)=P_3$) and that $\bR_3^*$ is optimal for (P4) (since $tr(\bR_3^*)=P_4$), i.e. $\bR_3^*=\bR^*_4$ and hence problems (P3) and (P4) have identical optimal values and optimal points (covariances). This establishes \eqref{eq.P3=P4}\footnote{We caution the reader that this does not imply that $\bR^*=\bR^*_3$ since $\bR^*$ is not necessarily unique and there may exist one with $tr(\bR^*)> P_3=P_4$.}.

To establish $P_3=P_4=P_0$ if $P_T\ge P_0$, observe that $C(P) < C(P_0)=C(P_T)=C_0$ for any $P<P_0$ (from \eqref{eq.P0}). Therefore, $P_3=P_4 \ge P_0$. Since $C(P_0)=C(P_T)=C(P_3)=C(P_4)$, it follows that $P_3=P_4=P_0$ and hence Algorithm 2 also solves (P3) and (P4).

If $P_T< P_0$ (which includes $P_0=\infty$), it follows from Proposition \ref{prop.muPT} that the TPC is active and, from Proposition \ref{prop.Properties}, that $C(P_T)$ is strictly increasing, which implies $P_3=P_4=P_T$ and $\bR^*_3=\bR^*_4=\bR^*=\bR'$, where the last equality implies that Algorithm 1 also solves (P3) and (P4).

\end{proof}

%============================================================================
\section{Numerical Experiments}
\label{sec.examples}

To validate the algorithms and demonstrate their performance, extensive numerical experiments have been carried out. We consider below some representative cases with 2 cooperative Evs and 2 PRs below, for both deterministic and randomly-generated channels.

\textbf{Example 1}: Fig. \ref{fig.r.step} illustrates the convergence of Algorithm 1 for the channel in \eqref{eq.NE.H}, i.e. the residual's Euclidian norm  $|\br(\bz_k)|$ versus the number $k$ of Newton steps for various values of $t$. Channel matrices $\bH_{1}$,  $\bH_{21}$,  $\bH_{22}$, $\bH_{31}$, $\bH_{32}$ are set as follows:
\bal\notag
\label{eq.NE.H}
&\bH_{1} =
\begin{bmatrix}
	0.32	& 0.66\\
	1.24	& 0.58
\end{bmatrix}\\ \notag
&\bH_{21} =
\begin{bmatrix}
	-0.58	& -1.15\\
	-0.37	& -1.07
\end{bmatrix},\
\bH_{22} =
\begin{bmatrix}
	0.17	& 0.73\\
	-0.07	& -0.54
\end{bmatrix}\\
&\bH_{31} =
\begin{bmatrix}
	1.47	& 0.32\\
	-1.57	& 0.01
\end{bmatrix}, \quad
\bH_{32} =
\begin{bmatrix}
	-0.83	& 0.38\\
	1.16	& -0.86
\end{bmatrix}
\eal
so that $\bW_{31}$ and $\bW_{32}$ are full rank (and hence the singularity is ruled out so that $\bR'=\bR^*$ in this case); the corresponding eigenvalues of $\bW_{1}-\bW_{2}$ are $(-2.53,1.16)$, i.e. the channel is non-degraded and "hard" for optimization (since the negative eigenmode is dominant). For all considered values of $t$, it takes only about 8 to 23 Newton steps to reach the machine precision level (around $10^{-12}$; recall that a globally-optimal point corresponds to $|\br|=0$). Also note the presence of two convergence phases: linear and quadratic. After the quadratic ("water-fall") phase is reached, the convergence is very fast. In general, higher values of $t$, which provide smaller gap to the capacity, require more steps to achieve the same precision, even though moderately so.

\begin{figure}[t]%[htbp]
\centerline{\includegraphics[width=3.3in]{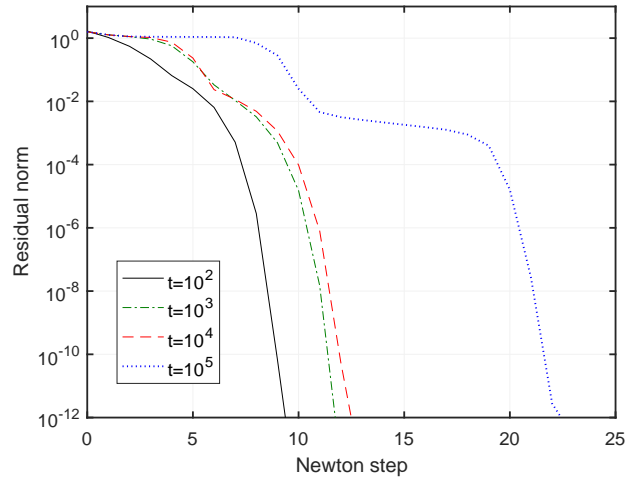}}
\caption{Convergence of the Newton method for different values of $t$; $P_T=5$ dB, $P_{I1}=P_{I2}=2$ dB, $\alpha=0.3, \beta=0.5$, channel matrices are as in \eqref{eq.NE.H}.}
\label{fig.r.step}
\end{figure}

\begin{figure}[t]%[htbp]
\centerline{\includegraphics[width=3.2in]{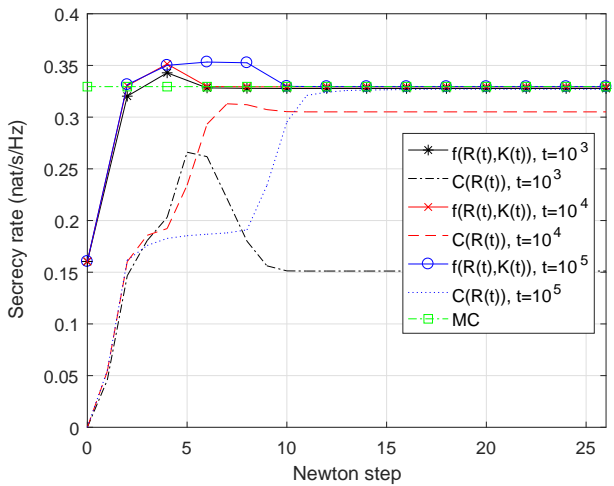}}
\caption{Achieved secrecy rate $C(\bR(t))$ and the upper bound $f(\bR(t),\bK(t))$ for the same setting as in Fig. \ref{fig.r.step}, via Algorithm 1 (for given $t$) and stochastic Monte-Carlo (MC) search.}
\label{fig.C.f.step}
\end{figure}

Fig. \ref{fig.C.f.step} shows the achieved secrecy rate $C(\bR(t))$ and its upper bound $f(\bR(t),\bK(t))$ under the same setting as in Fig. \ref{fig.r.step}. While the attained secrecy rate and the upper bound do converge to the secrecy capacity, the convergence is significantly non-monotonic here, unlike that in Fig. \ref{fig.r.step}, where the residual is decreasing monotonically. It takes more steps for  $C(\bR(t))$ to converge, as compared to $f(\bR(t),\bK(t))$, since the former is much more sensitive to $\bR$ than the latter (due to the strong negative eigenmode of $\bW_{1}-\bW_{2}$). While $t=10^3$ is sufficient to evaluate accurately the capacity via $f(\bR(t),\bK(t))$, it takes $t=10^5$ to get the same accuracy via $C(\bR(t))$ so we conclude that, in addition to being convex-concave in the right way, $f(\bR(t),\bK(t))$ is more robust (less sensitive) than $C(\bR)$: while using $t=10^3$ entails no visible loss in precision for the capacity estimate via $f(\bR(t),\bK(t))$, it induces a significant loss (of about 50\%) in attained secrecy rate $C(\bR(t))$.

For properly selected $t$, it takes a moderate number of 10 to 15 Newton steps for the algorithm to converge in terms of achieved secrecy rates. In general, the numerical complexity of Algorithm 1 follows that of the standard barrier method \cite{Boyd-04}: the number of Newton steps for each value of $t$ scales as
\bal
\label{eq.Nst}
O\left(\sqrt{n_v}\ln\frac{1}{\epsilon}\right)
\eal
while the number of flops for each Newton step scales as $O(n_v^3)$ so that the overall number of flops scales as $O(n_v^3\sqrt{n_v}\ln\epsilon^{-1})$, i.e. polynomially in $n_v$ and logarithmically in $\epsilon$, where $n_v$ is the number of independent variables of the max-min problem:
\bal
n_v =\frac{m(m + 1)}{2} + n_1\sum_j n_{2j}
\eal
This should be contrasted with the complexity of the non-convex max problem in \eqref{eq.secrecy}: if the global optimum is found using a generic non-convex optimization algorithm, the complexity scales as
\bal
O\left(\left(1/\epsilon\right)^{n_v'}\right),\ n_v'=m(m+1)/2,
\eal
i.e. exponentially large in $n_v'$, see e.g. \cite{Vavasis-95}\cite{Nesterov-04}, -- a stark contrast to \eqref{eq.Nst}, which scales sub-polynomially in $n_v$.

We further remark that while higher values of $t$ result in higher precision for the achieved secrecy rate $C(\bR(t))$, they also require more steps to reach the same precision level $\epsilon$, even though this number is still moderate.

Algorithm 1 was further validated by comparing its achieved secrecy rates with those attained by extensive stochastic Monte-Carlo (MC) search (where a large number, e.g. $10^5$, of covariance matrices were randomly generated within the feasible set and the best one was selected as an optimal covariance). As Fig. \ref{fig.C.f.step} shows, these two methods agree well with each other.

\textbf{Example 2}: To further validate Algorithm 1, its performance was evaluated on 200 randomly-generated channels (from i.i.d. $N(0,1)$ distribution for each entry of each channel) with different numbers of antennas. No significant difference with Fig. \ref{fig.r.step} and \ref{fig.C.f.step} was found. The number of Newton steps to reach the same precision is somewhat larger for larger number of antennas but only moderately so, in agreement with \eqref{eq.Nst}. Fig. \ref{fig.R.rand} shows the averaged secrecy rates found via Algorithm 1 (i.e. $C(\bR')$ and $f(\bR',\bK')$) as well as via extensive MC search, for different numbers of antennas at each node, with 2 Evs and 2 PRs present at each setting. Clearly, the results of Algorithm 1 and MC search agree well with each other, thus validating Algorithm 1.

\begin{figure}[t]%[htbp]
\centerline{\includegraphics[width=3.3in]{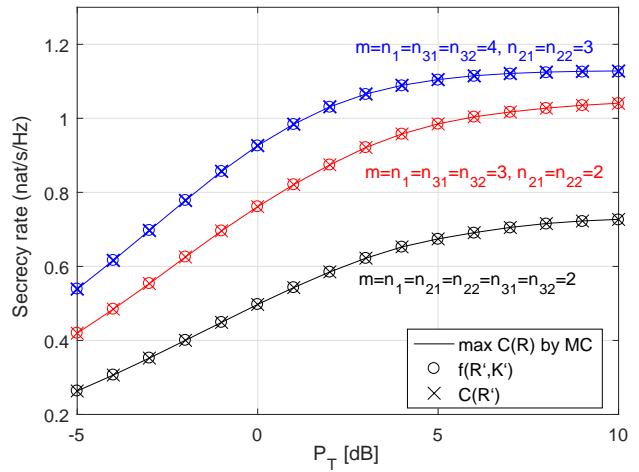}}
\caption{Achieved secrecy rates $C(\bR')$ and $f(\bR',\bK')$ via Algorithm 1 and extensive MC search, averaged over 200 randomly-generated channels with 2 Evs and 2 PRs; $P_{I1}=P_{I2}= 2$ dB, $\epsilon=10^{-8},\ \delta=10^{-3}$, $\alpha=0.3,\ \beta=0.5,\ \eta=5,\ t_{min}=10^2,\ t_{max}=10^5$.}
\label{fig.R.rand}
\end{figure}

\textbf{Example 3}: In this example, we demonstrate that while
\bal
\underset{\bR\in S_{\bR}}{\max}C(\bR) = \underset{\bR\in S_{\bR}}{\max}\ \underset{\bK\in S_{\bK}}{\min}f(\bR,\bK)
\eal
as Theorem 1 indicates, $\bR'$, which is a maximizer of $f(\bR,\bK)$, may not be a maximizer of $C(\bR)$, i.e. it is not necessarily an optimal covariance $\bR^*$ attaining the secrecy capacity, $\bR^* \neq \bR'$, as discussed above. In fact, the difference can be quite significant. Let the channel matrices be as follows: $\bH_1$, $\bH_{21}$, $\bH_{22}$ as in \eqref{eq.NE.H}, and set
\bal
 \label{eq.NE.H2}
\bH_{31} =
 \begin{bmatrix}
 	-0.23	& -0.16\\
 	 -0.05	&  -0.71
 \end{bmatrix},\
\bH_{32} =
 \begin{bmatrix}
 	-0.47	& 0.09
 \end{bmatrix}
 \eal
so that $\bW_{31}$ is full rank and $\bW_{32}$ is rank deficient and hence the singular case is possible. It is clear from Fig. \ref{fig.R*=not.R'} that $\bR'$ is not a maximizer of $C(\bR)$ in the singular case, i.e. for $P_T$ larger than about 11 dB (which corresponds to inactive TPC and the active IPC channel matrix being singular) where $C(\bR')$ drops down significantly while $f(\bR',\bK')$ returned by Algorithm 1 is always a good estimate of the capacity. Hence, $\bR'$ cannot be used for optimal signaling as an optimal covariance under inactive TPC (this would entail about 70\% loss in achieved secrecy rate). Note also that Algorithm 2 is able to find an optimal covariance even in the singular case and its attained secrecy rate agrees well with that of extensive Monte-Carlo search and that via $f(\bR',\bK')$.

To demonstrate that sub-optimal algorithms (based on 1st order Taylor expansion of the non-convex part of the max problem) may get trapped at a local optimum that is far away from the global one, the secrecy rate maximization algorithm in \cite{Cumanan-14}, which is based on this strategy, was implemented using the popular convex optimization toolbox CVX \cite{CVX-18} in each iteration. As Fig. \ref{fig.R*=not.R'} shows, it does get trapped in a local optimum (close to 0), far away from the global one, as expected from the discussion in the Introduction. Hence, Taylor-based sub-optimal algorithms should be used with caution (or avoided at all) when the original problem is not convex (as in this example, where the channel is not degraded and negative eigenmode dominates, making it "hard" for optimization).

\begin{figure}[t]%[htbp]
	\centerline{\includegraphics[width=3.3in]{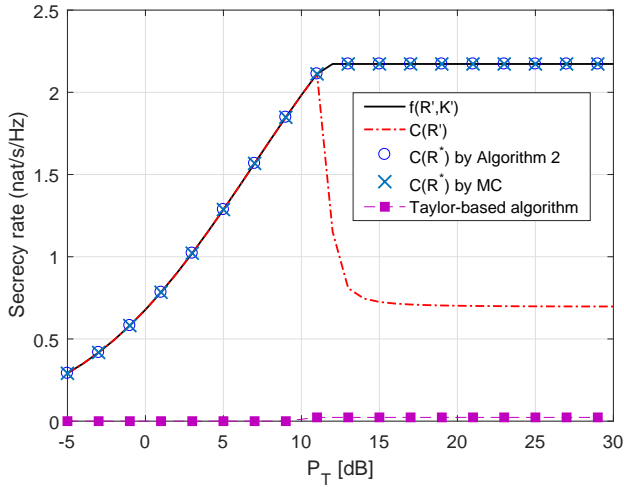}}
	\caption{Secrecy rate as a function of the TPC power $P_T$ for $\bH_1$, $\bH_{21}$, $\bH_{22}$ as in \eqref{eq.NE.H} while $\bH_{31}$ and $\bH_{32}$ are as in  \eqref{eq.NE.H2}, $P_{I1}=P_{I2}=5$ dB. Note that $\bR'$ is not always a maximizer of $C(\bR)$ but the secrecy rate attained by Algorithm 2 always agrees well with that of MC and $f(\bR',\bK')$. Taylor-based algorithm gets trapped at a local optimum, far away from the global one, since the channel is not degraded.}
	\label{fig.R*=not.R'}
\end{figure}

\begin{figure}[t]%[htbp]
	\centerline{\includegraphics[width=3.3in]{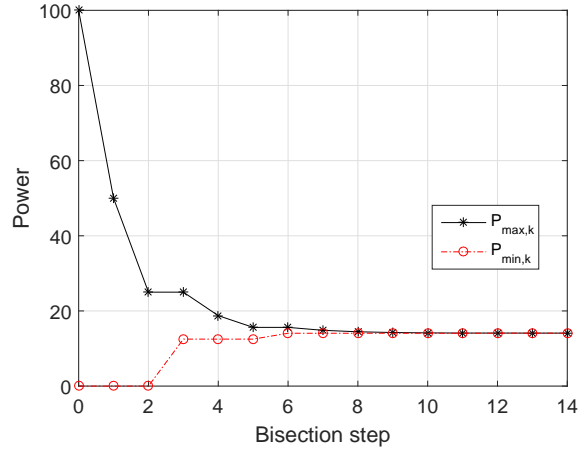}}
	\caption{Convergence  of Algorithm 2 for the setting of Fig. \ref{fig.R*=not.R'}; $P_T=100$ (20 dB), $P_{I1} = P_{I2}= 3.16$ (5 dB); $\epsilon=\delta= 10^{-4}$; estimated $P_0=14.09$. Note that the difference between $P_{max,k}$ and $P_{min,k}$ decreases sharply with $k$ and both sequences monotonically converge to $P_0$, as expected from the analysis.}
	\label{Fig.A2.cov}
\end{figure}

\begin{figure}[t]%[htbp]
	\centerline{\includegraphics[width=3.3in]{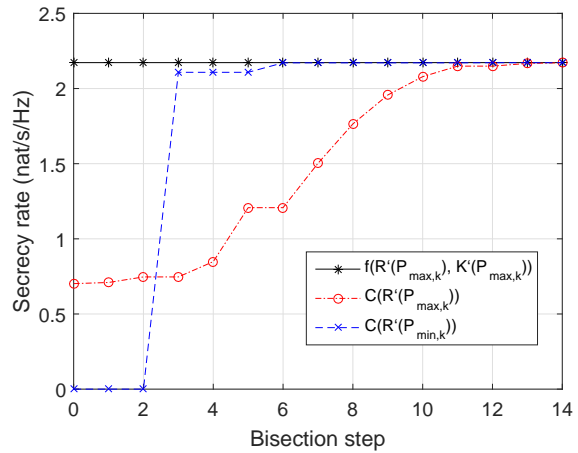}}
	\caption{Convergence  of secrecy rates from Algorithm 2 for the setting of Fig. \ref{Fig.A2.cov}. Note that both $C(\bR'(P_{min,k}))$ and $C(\bR'(P_{max,k}))$ converge to the capacity, where the former converges faster while $f(\bR'(P_{max,k}), \bK'(P_{max,k})) =C$ for all $k$, as expected. }
	\label{Fig.A2.C.cov}
\end{figure}

Fig. \ref{Fig.A2.cov} demonstrates the convergence of Algorithm 2 for the setting of Fig. \ref{fig.R*=not.R'} with $P_T=100$ (20 dB), $P_{I1} = P_{I2}= 3.16$ (5 dB); $\epsilon=\delta= 10^{-4}$. Note that the convergence is exponentially fast so that only a few bisection steps is needed and that both sequences converge monotonically to $P_0=14.09$. Fig. \ref{Fig.A2.C.cov} shows the convergence of attained secrecy rates $C(\bR'(P_{min,k}))$ and $C(\bR'(P_{max,k}))$ for the same setting. Note that $C(\bR'(P_{max,k}))$ is initially significantly below the capacity while $f(\bR'(P_{max,k}), \bK'(P_{max,k})) =C =2.17$ for all $k$, confirming our earlier observation that $\bR'$ is not necessarily a maximizer of $C(\bR)$ while $f(\bR',\bK')=C$ always holds. Algorithm 2 overcomes this problem by properly reducing the TPC power $P_T$ to make the TPC active. In this process, both $C(\bR'(P_{min,k}))$ and $C(\bR'(P_{max,k}))$ converge to the capacity, where the former converges faster (it takes only 6 iterations) while $f(\bR'(P_{max,k}), \bK'(P_{max,k})) =C$ for all $k$, as expected from the analysis. As a by-product, the minimum Tx power $P_0=14.09$ needed to achieve the secrecy capacity $C =2.17$ is also determined by Algorithm 2. Note that this power is significantly smaller than the TPC power $P_T=100$, hence allowing significant power savings.

\textbf{Example 4}: To further validate Algorithm 2, we compare its performance with that of MC search for 200 randomly-generated channels for which the singularity condition of Section \ref{sec.algorithm2} is satisfied, with $m=n_1=4, n_{21}=n_{22}=3, n_{31}= n_{32}= 2$. Fig. \ref{fig.R*.A2.rand} shows the respective averaged secrecy rates. Note that the results of Algorithm 2 and MC search agree well with each other while using $\bR'$ from Algorithm 1 alone does not always maximize $C(\bR)$, in agreement with Section \ref{sec.algorithm2} and Example 3; the gap is significant at high $P_T$ (e.g. 10 dB). This shows a dramatic impact of IPC on algorithm's performance. This figure, along with Fig. \ref{fig.R*=not.R'}, also shows that Algorithm 1 alone is not sufficient and Algorithm 2 is really needed to find an optimal covariance matrix in the singular case. On the contrary, Algorithm 1 is sufficient to find the capacity in the general case via $f(\bR',\bK')$.

\begin{figure}[t]%[htbp]
	\centerline{\includegraphics[width=3.3in]{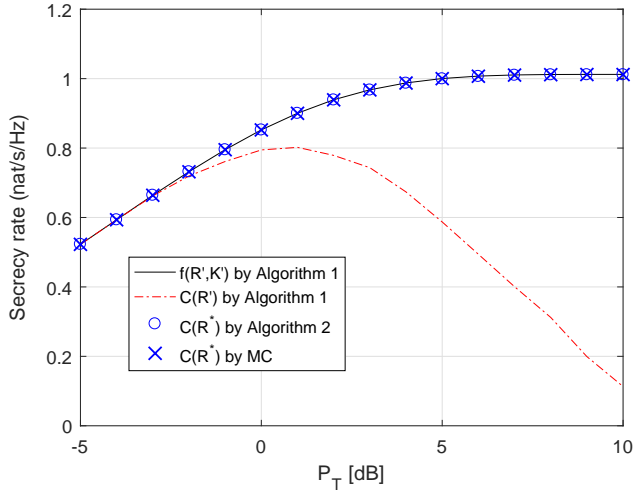}}
	\caption{Secrecy rate vs. the TPC power $P_T$, averaged over 200 random realizations of singular channels; $P_{I1}=P_{I2}=-2$ dB, $\epsilon=10^{-4},\ \delta=10^{-3}$, $\alpha=0.3,\ \beta=0.5,\ \eta=5,\ t_{min}=10^2,\ t_{max}=10^5$. Note that $\bR'$ is not always a maximizer of $C(\bR)$ but the secrecy rate attained by Algorithm 2 always agrees well with that of MC search and $f(\bR',\bK')$.}
	\label{fig.R*.A2.rand}
\end{figure}

%\vspace*{-1\baselineskip}
%================================================================================
\section{Conclusion}

Optimal secure signaling over multi-user MIMO wiretap channels has been studied under interference constraints (e.g. as in CR) in this paper. While several algorithms have been presented in the literature for secrecy rate maximization in this setting, they are either limited to the MISO setting (single-antenna receiver) or suffer from the lack of provable convergence to a global optimum and may get trapped in a local optimum far away from the global one.

In this paper, we presented two algorithms  for \textit{global} secrecy rate maximization under interference constraints in the full multi-user MIMO setting with provable (global) convergence to the secrecy capacity. These algorithms avoid using approximation-based approach (as in all known algorithms) and hence avoid the danger of being trapped in a local optimum (or stationary point) far away from the global one. This is accomplished by using the recent secrecy capacity characterization in the interference-constrained setting as a max-min problem, where both problems are convex.
As a by-product, the minimum transmit power needed to achieve the secrecy capacity is also determined via Algorithm 2. This algorithm also solves the dual problem of globally minimizing the total Tx power subject to the secrecy rate constraint, in addition to the IPCs. Numerical experiments validate the convergence analysis and demonstrate fast convergence for both algorithms as well as their superiority to the sub-optimal algorithms known in the literature.

Finally, we remark that these algorithms can also be used to evaluate the secrecy capacity and globally-optimal signaling strategy under per-antenna power constraint, in addition to or instead of the TPC. This can be accomplished by setting some $\bW_{3j}$ to be diagonal matrices with 0-1 entries.

%\vspace*{-1\baselineskip}
%======================================================================
%\section{Acknowledgement}

%The authors are thankful to Dr. V.I. Mordachev and E.V. Sinkevich (BSUIR, Minsk, Belarus) for their support and encouragement.

%\vspace*{-.5\baselineskip}
%\newpage
%========================================================================
\section{Appendix}

%\vspace*{-.5\baselineskip}
\subsection{Gradients and Hessians}
While gradients and Hessians can be computed numerically via finite differences, this results in lower efficiency in addition to a possible loss in precision (due to numerical "noise"), which affects negatively convergence of the algorithm. Hence, we provide below analytical expressions for gradients and Hessians obtained, after some manipulations, using the standard rules of matrix differential calculus (see e.g. \cite{Harville-97}\cite{Magnus}):
\bal
\notag
\nabla_{\bx}f_t &= \bD_{m}^T vec(\bZ_{1}-\bZ_{2}+t^{-1}\bR^{-1})\\ \notag &-\frac{1}{t}g_1(\bR) veh(\bI) -\frac{1}{t}\sum_{j} g_{3j}(\bR)\bw_j,\\
\nabla_{\by}f_t &= \widetilde{\bD}_{n}^T vec((\bK+\bQ)^{-1}-(1+t^{-1})\bK^{-1})
\eal
where $\bQ=\bH\bR\bH^T$, $\bw_j=veh(2\bW_{3j}-diag(\bW_{3j}))$,
\bal\notag
&g_1(\bR)=(P_{T}-tr(\bR))^{-1},\\ \notag
&g_{3j}(\bR)= (P_{Ij} -tr(\bW_{3j}\bR))^{-1},\\ \notag
&\bZ_{1}=(\bI+\bH^T\bK^{-1}\bH\bR)^{-1}\bH^T\bK^{-1}\bH,\\
&\bZ_{2}=(\bI+\bW_{2}\bR)^{-1}\bW_{2},
\eal
and $\bD_{m}$ is the $m^2 \times m(m+1)/2$  duplication matrix defined from $vec(\bR)=\bD_{m}veh(\bR)$ \cite{Magnus}, $\widetilde{\bD}_{n}$ is $(n_{1}+n_{2})^2 \times n_{1}n_{2}$ reduced duplication matrix defined from $d\bk = \widetilde{\bD}_n d\tilde{\bk}$, where
\bal\notag
d\bk = vec(d\bK),\ d\tilde{\bk}=vec(d\bN), \
d\bK = \left(
\begin{array}{cc}
	\bo & d\bN^T \\
	d\bN & \bo \\
\end{array}
\right)
\eal
Likewise, the Hessians are
\bal
\label{eq.Hessian.xx}
\notag
\nabla_{\bx\bx}^{2}f_t  &=-\bD_{m}^T(\bZ_{1}\otimes \bZ_{1}-\bZ_{2}\otimes \bZ_{2}+t^{-1}\bR^{-1}\otimes \bR^{-1})\bD_{m}\\ \notag
 &-\frac{1}{t} g_1(\bR)^{2} veh(\bI) veh(\bI)^T -\frac{1}{t}\sum_{j}g_{3j}(\bR)^{2}\bw_j \bw_j^T\\ \notag
\nabla_{\bx\by}^{2}f_t &= -\bD_{m}^T(\bH^T(\bK+\bQ)^{-1}\otimes \bH^T(\bK+\bQ)^{-1})\widetilde{\bD}_{n}\\ \notag
\nabla_{\by\by}^{2}f_t &= \widetilde{\bD}_{n}^T (-(\bK+\bQ)^{-1}\otimes (\bK+\bQ)^{-1}\\
&+(1+t^{-1})\bK^{-1}\otimes \bK^{-1})\widetilde{\bD}_{n}
\eal
%While detailed derivations are omitted here due to the page limit, they are available in \cite{Dong-18b}.

%\vspace*{-1\baselineskip}
%============================================================================
\subsection{Proof of Proposition \ref{prop.non-singular.KKT}}

The proof is based on the following two Lemmas.
\begin{lemma}
	\label{lemma.partial.hessian}
	Partial Hessian $\nabla^2_{\bx\bx} f_t,\ \nabla^2_{\by\by} f_t$ are non-singular for each $t>0$, $\bR \in S'_{\bR},\bK \in S'_{\bK}$.
\end{lemma}
\begin{proof}
First, since $\bQ \geq \bo$ and $\bK>0$, then $\bK+\bQ \ge  \bK >\bo$ so that
 \bal
 (\bK+\bQ)^{-1} \leq \bK^{-1}
\eal
and, using the  properties of Kronecker products \cite{Zhang},
  \bal
   \bK^{-1}\otimes \bK^{-1} \geq (\bK+\bQ)^{-1}\otimes (\bK+\bQ)^{-1}
\eal
so that
\bal
\label{eq.L1.3}
\notag
	(1+t^{-1}) \bK^{-1}&\otimes \bK^{-1} -(\bK+\bQ)^{-1}\otimes (\bK+\bQ)^{-1}\\
	&\geq t^{-1}\bK^{-1}\otimes \bK^{-1} > \bo.
\eal
Secondly, since $\widetilde{\bD}_{n}$ is of full column rank, it follows that $ \widetilde{\bD}_{n}\by\neq \bo$ for any $\by\neq \bo$, so that
\bal\notag
 \by^T\nabla_{\by\by}^2 f_t\by &=\by^T\widetilde{\bD}^T_{n}((1+t^{-1}) \bK^{-1}\otimes \bK^{-1}\\
  &-(\bK+\bQ)^{-1}\otimes (\bK+\bQ)^{-1})\widetilde{\bD}_{n}\by>0
\eal
where the inequality is due to \eqref{eq.L1.3}, and thus $\nabla_{\by\by}^2 f_t>\bo$, as required.

To prove the non-singularity of $\nabla_{\bx\bx}^2 f_t$, first note that $\bH^T\bK^{-1}\bH \geq \bW_2$ so that $(\bH^T\bK^{-1}\bH)^{-1}\leq \bW_2^{-1}$ and hence $\bZ_1 \geq \bZ_2$ as follows:
\bal
\notag
\bZ_1&=(\bI+\bH^T\bK^{-1}\bH\bR)^{-1}\bH^T\bK^{-1}\bH\\
\notag
     &=(\bR+(\bH^T\bK^{-1}\bH)^{-1})^{-1}\\
\notag
     &\geq (\bR+\bW_2^{-1})^{-1}\\
     &=(\bI+\bW_2\bR)^{-1}\bW_2=\bZ_2.
\eal
The case of singular $\bH^T\bK^{-1}\bH$ and $\bW_2$ can be considered using the standard continuity argument (see e.g. \cite{Zhang}). Using this equality and the property of Kronecker products, it follows that $\bZ_1\otimes\bZ_1 \geq \bZ_2\otimes\bZ_2 $, and, since $t^{-1}\bR^{-1}\otimes\bR^{-1}>\bo$,
	\bal
	\bZ_1\otimes\bZ_1 - \bZ_2\otimes\bZ_2 + t^{-1}\bR^{-1}\otimes\bR^{-1} > \bo
	\eal
Since $\bD_{m}$ is of full column rank \cite{Magnus}, $\bD_{m}\by\neq \bo$  for any $\by\neq \bo$, so that
\bal\notag
\by^T\bD^T_{m}(\bZ_1\otimes\bZ_1 - \bZ_2\otimes\bZ_2 + t^{-1}\bR^{-1}\otimes\bR^{-1})\bD_{m}\by > 0
\eal
and hence
\bal
\bD^T_{m}(\bZ_1\otimes\bZ_1 - \bZ_2\otimes\bZ_2 + t^{-1}\bR^{-1}\otimes\bR^{-1})\bD_{m} > \bo
\eal
Applying all these inequalities to \eqref{eq.Hessian.xx}, one obtains $\nabla_{\bx\bx}^{2}f_t<\bo$, as desired.
\end{proof}

\begin{lemma}
\label{lemma.hessian}
The Hessian
	\bal
	D\br =
	\left[
	\begin{array}{cc}
		-\bT_{11} & \bT_{12}\\
		\bT_{21} & \bT_{22}\\
	\end{array}
	\right]
	\eal
is non-singular if partial Hessians $\bT_{11}, \bT_{22}$ are non-singular, i.e. if $\bT_{11}, \bT_{22} > \bo $, where $\bT_{11}=-\nabla^2_{xx} f_t,\ \bT_{12}=\nabla^2_{xy} f_t,\ \bT_{21}=\bT_{12}^T=\nabla^2_{yx} f_t,\ \bT_{22}=\nabla^2_{yy} f_t$.
\end{lemma}
\begin{proof}
Note that $D\br$ is a square $n_T\times n_T$ matrix, where $n_T = m(1+m)/2+n_1\sum_{i=1}^{N}n_{2i}$, so that proving its non-singularity is equivalent to proving that $|D\br|\neq 0$. Using an expression for the determinant of a partitioned matrix \cite{Zhang}, we have
\bal
\notag
|D\br|&= |-\bT_{11}| |\bT_{22}+\bT_{21}\bT_{11}^{-1}\bT_{12}|\\
\label{eq.det.T}
&=(-1)^{n_T} |\bT_{11}| |\bT_{22}+\bT_{12}^T\bT_{11}^{-1}\bT_{12}|
\eal
where \eqref{eq.det.T} follows since $\bT_{12}^T=\bT_{21}$. According to Lemma \ref{lemma.partial.hessian}, $\bT_{11}, \bT_{22}>0$, so that $|\bT_{11}|>0$ and $\bT_{22}+\bT_{12}^T\bT_{11}^{-1}\bT_{12}> \bo$ and hence $|\bT_{22}+\bT_{12}^T\bT_{11}^{-1}\bT_{12}|>0$, from which  $|D\br|\neq  0$ follows.
\end{proof}

Combining Lemma \ref{lemma.partial.hessian} and Lemma \ref{lemma.hessian}, we conclude that the Hessian $D\br$ is non-singular at each step of the Newton method.

%\vspace*{-1\baselineskip}
%============================================================================
\subsection{Proof of Proposition \ref{prop.Properties}}

Let $\sS_R(P_T)$ be the feasible set for a given TPC power $P_T$. 1st part of Property 1 follows from the fact that if $\bR\in \sS_R(P_T)$, then $\bR\in \sS_R(P_1)$ for any $P_1 \ge P_T$, i.e. if $\bR$ is feasible for TPC power $P_T$, then it is also feasible for any higher TPC power $P_1$. 2nd part follows from Property 4.

The concavity can be proved by contradiction. Assume that $C(P_T)$ is not concave, i.e. there exist powers $P_1,\ P_2$ and $0< \theta <1$ such that
\bal
\label{eq.PT.sharing}
C(\theta P_1 + (1-\theta) P_2) < \theta C(P_1) + (1-\theta) C(P_2)
\eal
Now, consider power/time sharing between power levels $P_1$ and $P_2$, i.e. transmitting under TPC power $P_1$ for $\theta$ fraction of time and under TPC power $P_2$ for $1-\theta$ fraction of time, so that the average Tx power does not exceed $P_a = \theta P_1 + (1-\theta) P_2$.
Let $\bR_k$ be an optimal Tx covariance under the TPC power $P_k$, $k=1,2$. Note that $\bR_k \in \sS_R(P_k)$ implies $\theta \bR_1  + (1-\theta) \bR_2 \in \sS_R(P_a)$, i.e. this power/time sharing is feasible under the TPC power $P_a$  and it achieves the secrecy rate equal to the right hand side of \eqref{eq.PT.sharing}, so that, from \eqref{eq.PT.sharing},  $C(P_a = \theta P_1 + (1-\theta) P_2)$ is not the capacity - a contradiction. We remark that using the standard optimization-based proof, as in e.g. \cite{Boyd-04} (Exercise 5.32), is not possible here since $C(\bR)$ is not concave (unless the channel is degraded). Continuity of $C(P_T)$ follows from its concavity \cite{Rockafellar-70}.

To prove Property 3, observe that
\bal
C(P)\le C(P_2)\le C(P_1)
\eal
for any $P\le P_2\le P_1$, since $C(P)$ is non-decreasing, and hence
\bal
C(P)= C(P_2)= C(P_1)
\eal
if $C(P) = C(P_1)$. It follows that $C(P)'_+=0$ and thus $C(P_1)'_+=0$ for any $P_1 \ge P$, since $C(P)'_+$ is non-increasing (since $C(P)$ is concave)    so that, from Corollary 24.2.1 in \cite{Rockafellar-70},
\bal
C(P_1) = C(P) + \int_{P}^{P_1} C(p)'_+dp = C(P)
\eal
for any $P_1 \ge P$. 2nd statement follows from 1st one.

To prove Property 4, use Corollary 24.2.1 in \cite{Rockafellar-70} again,
\bal
C(P_T) = C(P_1) + \int_{P_1}^{P_T} C(p)'_-dp > C(P_1)
\eal
from which it follows that $C(p)'_->0$ for some $P_1 \le p \le P_T$ and hence $C(p)'_-\ge C(P_1)'_- > 0$ for any $p \le P_1$ so that
\bal
C(P_1) = C(P_2) + \int_{P_2}^{P_1} C(p)'_-dp > C(P_2)
\eal
2nd part of this property follows from the fact that $C(P)'_-$ is non-increasing (since $C(P)$ is concave).

%\vspace*{-0.5\baselineskip}
%============================================================================
\subsection{Proof of Proposition \ref{prop.muPT}}

The proof is by contradiction as follows. First, note that $C(P_T)$ is an optimal value of the max-min problem (P1) in \eqref{eq.secrecy}. Consider now the same problem but without the TPC (under the IPC only):
\bal
\label{eq.P3}
\textrm{(P3)}: \  \max_{\bR\in \sS_R(\infty)} \min_{\bK \in \sS_K}  f(\bR,\bK)
\eal
for which the respective KKT conditions are
\bal
\label{eq.gradient.2}
&\nabla_{\bR} f(\bR,\bK) + \bM_1 - \sum_{j=1}^{K}\mu_j\bW_{3j}= \bo, \\
\label{eq.complementary slackness.2}
&\bM_1\bR = \bo,\  \mu_{j} (tr(\bW_{3j}\bR)-P_{Ij})=0,\\
\label{eq.primal.dual.2}
&tr(\bW_{3j}\bR) \le P_{Ij},\ \bR, \bM_1 \ge \bo, \ \mu_{j} \ge 0, \\
\label{eq.gradient.K.2}
&\nabla_{\bK} f(\bR,\bK) - \bM_2 + \bLam = \bo, \\
\label{eq.slackness.K.2}
&\bM_2\bK = \bo,\ \bK, \bM_2 \ge \bo
\eal
where \eqref{eq.gradient.2}-\eqref{eq.primal.dual.2} and \eqref{eq.gradient.K.2}-\eqref{eq.slackness.K.2} are the KKT conditions for maximization over $\bR$ and minimization over $\bK$, respectively; $\bM_{1(2)} \ge 0$ are (matrix) Lagrange multipliers responsible for $\bR\ge 0$ and $\bK \ge 0$ constraints; $\bLam$ is Lagrange multiplier responsible for the equality constraint in \eqref{eq.set.K};  $\mu_j \ge 0$ is Lagrange multiplier responsible for $j$-th IPC; \eqref{eq.gradient.2} and \eqref{eq.gradient.K.2} are the stationarity conditions with respect to $\bR$ and $\bK$; the equalities in \eqref{eq.complementary slackness.2} and \eqref{eq.slackness.K.2} are the complementary slackness conditions while  the inequalities in \eqref{eq.primal.dual.2} and \eqref{eq.slackness.K.2} are the primal and dual feasibility constraints.

On the other hand, the KKT conditions of the original problem (P1) are
\bal
\label{eq.gradient}
&\nabla_{\bR} f(\bR,\bK) + \bM_1 - \mu\bI\ - \sum_{j}\mu_j\bW_{3j}= \bo, \\
\label{eq.complementary slackness}
&\bM_1\bR = \bo,\ \mu (tr(\bR)-P_T)=0,\ \mu_{j} (tr(\bW_{3j}\bR)-P_{Ij})=0,\\
\label{eq.primal.dual}
&tr(\bR) \le P_T,\ tr(\bW_{3j}\bR) \le P_{Ij},\ \bR, \bM_1 \ge \bo, \ \mu, \mu_{j} \ge 0\\
\label{eq.gradient.K}
&\nabla_{\bK} f(\bR,\bK) - \bM_2 + \bLam = \bo, \\
\label{eq.slackness.K}
&\bM_2\bK = \bo,\ \bK, \bM_2 \ge \bo
\eal
where $\mu\ge 0$ is Lagrange multiplier responsible for the TPC.

Now note that any solution of the KKT conditions of (P1) in \eqref{eq.gradient}-\eqref{eq.slackness.K} with $\mu(P_T)=0$ also solves \eqref{eq.gradient.2}-\eqref{eq.slackness.K.2} and hence (P3) (since the KKT conditions are sufficient for optimality in both cases). Thus, if $\mu(P_T)=0$, then
\bal
C(P_T)=C(\infty) \ge C(P_0) > C(P_T)
\eal
i.e. a contradiction, where the last inequality is from Property 1 in Proposition \ref{prop.Properties}. Hence, we conclude that $\mu(P_T)>0$ for any $P_T < P_0$. Once can further show that $\mu(P_T)=0$ for any $P_T > P_0$: observe that $C(P_T)'_-=C(P_T)'_+=C(P_T)'=0$ if $P_T > P_0$ (since $C(P_T)=C(P_0)$ for any $P_T > P_0$ from the definition of $P_0$) and hence $\mu(P_T)=C(P_T)'=0$.

%\vspace*{-1\baselineskip}
%======================================================================

\end{document}